\documentclass[11pt]{article}
\usepackage[english]{babel}
\usepackage{amsmath}
\usepackage{amsthm}
\usepackage{amssymb}
\usepackage{amsfonts}
\usepackage{mathtools}
\usepackage{mathrsfs}
\usepackage{thmtools}
\usepackage{algorithm}
\usepackage{algorithmic}
\usepackage{pdfsync}
\usepackage{graphicx}
\usepackage{fullpage}
\usepackage{bbm}
\usepackage{tabto}
\usepackage[numbers]{natbib}
\usepackage[pagebackref,
    colorlinks=true,
    urlcolor=blue,
    linkcolor=blue,
    citecolor=blue,
]{hyperref}
\usepackage[nameinlink]{cleveref}

\usepackage[inline,marginclue,index]{fixme}  %Remove draft to erase all notes, remove inline to erase inline notes.

\fxsetup{theme=color,mode=multiuser}
\FXRegisterAuthor{josh}{ajosh}{\colorbox{orange!90!yellow}{\color{white}Josh}}  %First two arguments must be different
\FXRegisterAuthor{linda}{alinda}{\colorbox{blue}{\color{white}Linda}}  %First two arguments must be different
\FXRegisterAuthor{matt}{amatt}{\colorbox{brown}{\color{white}Matt}}

\newtheorem{thm}{Theorem}[section]

\newtheorem{lemma}[thm]{Lemma}

\newtheorem{definition}[thm]{Definition}
\newtheorem{corollary}[thm]{Corollary}
\newtheorem{claim}[thm]{Claim}

\newtheorem{example}{Example}

\newcommand{\E}{\mathbb{E}}

\DeclarePairedDelimiter{\ceil}{\lceil}{\rceil}
\DeclarePairedDelimiter{\set}{\{}{\}}

\DeclarePairedDelimiter{\norm}{\lVert}{\rVert}
\DeclarePairedDelimiter{\sq}{[}{]}
\DeclarePairedDelimiter{\paren}{\lparen}{\rparen}
\DeclareMathOperator*{\argmin}{argmin}
\DeclarePairedDelimiterX{\cond}[2]{[}{]}{#1\,\delimsize\vert\,\mathopen{} #2}

\newcommand{\wsigma}{\widetilde{\sigma}}
\newcommand{\dist}{\mathcal{F}}

\newcommand{\fbias}{feature-amplified bias}
\newcommand{\pratio}{prophet utility ratio}
\newcommand{\oratio}{online utility ratio}

\usepackage{booktabs}
\title{Optimal Stopping with Multi-Dimensional Comparative Loss Aversion}
\author{
Linda Cai\footnote{Princeton University, \texttt{tcai@princeton.edu}.} \and Joshua Gardner\footnote{Princeton University, \texttt{jgardner2051@gmail.com}.} \and S. Matthew Weinberg\footnote{Princeton University, \texttt{smweinberg@princeton.edu}. Supported by NSF CAREER CCF-1942497 and a Google Research Award.}}
\date{}
\begin{document}

\maketitle
\begin{abstract}
Motivated by behavioral biases in human decision makers, recent work by \citet{kleinberg2021optimal} explores the effects of loss aversion and reference dependence on the prophet inequality problem, where an online decision maker sees candidates one by one in sequence and must decide immediately whether to select the current candidate or forego it and lose it forever. In their model, the online decision-maker forms a reference point equal to the best candidate previously rejected, and the decision-maker suffers from loss aversion based on the quality of their reference point, and a parameter $\lambda$ that quantifies their loss aversion. We consider the same prophet inequality setup, but with candidates that have \emph{multiple features}. The decision maker still forms a reference point, and still suffers loss aversion in comparison to their reference point as a function of $\lambda$, but now their reference point is a (hypothetical) combination of the best candidate seen so far \emph{in each feature}. 
 
Despite having the same basic prophet inequality setup and model of loss aversion, conclusions in our multi-dimensional model differs considerably from the one-dimensional model of~\citet{kleinberg2021optimal}. For example,~\citet{kleinberg2021optimal} gives a tight closed-form on the competitive ratio that an online decision-maker can achieve as a function of $\lambda$, for any $\lambda \geq 0$. In our multi-dimensional model, there is a sharp phase transition: if $k$ denotes the number of dimensions, then when $\lambda \cdot (k-1) \geq 1$, \emph{no non-trivial competitive ratio is possible}. On the other hand, when $\lambda \cdot (k-1) < 1$, we give a tight bound on the achievable competitive ratio (similar to~\cite{kleinberg2021optimal}). As another example,~\cite{kleinberg2021optimal} uncovers an exponential improvement in their competitive ratio for the random-order vs.~worst-case prophet inequality problem. In our model with $k\geq 2$ dimensions, the gap is at most a constant-factor. We uncover several additional key differences in the multi- and single-dimensional models.
\end{abstract}

\newpage

% Paper body
\section{Introduction}

Reference dependence describes the tendency for human decision-makers to compare options against a previously seen or anticipated reference point, such that an option inferior to the reference point is perceived as being of a lower quality than it would be had the person simply not been aware of the reference point \cite{tversky1991loss}. Similarly, loss aversion describes the tendency of human decision-makers to be more sensitive to losses than to gains of equivalent magnitude.  Both phenomena have robust empirical support~\cite{kahneman2013prospect}, and have been observed in the context of online decision-making \cite{schunk2009relationship}. 

Recent work of~\citet{kleinberg2021optimal}, has considered these behavioral biases in the context of optimal stopping, constructing biased agents who suffer from loss aversion in comparing their chosen candidate against previously seen and foregone candidates. Prophet inequalities are a canonical optimal stopping problem, where an agent sees $n$ candidates, each with a scalar value drawn independently from known distributions. The candidates are presented online, such that whenever the agent sees a candidate they must either accept it, halting their search, or forego it and thus lose the candidate forever. In the event they select a candidate, the utility they receive is exactly the value of the candidate. Motivated by the role of reference points and loss aversion,~\cite{kleinberg2021optimal} pose a modified prophet inequality where the utility received by the agent is discounted based on their reference point. Specifically, the biased agent also remembers the best candidate they have previously seen, and perceives new candidates as having decreased utility if the best candidate previously seen had a higher value. 

The elegant model of \citet{kleinberg2021optimal} captures scenarios in which candidates are well-represented by a single scalar (such as monetary value). Indeed, their model is partially motivated by the empirical findings of \citet{schunk2009relationship} concerning the unwillingness of investors to sell stocks below the price at which they were purchased. Similarly, their model captures the phenomenon of biased decision-makers stopping much sooner than optimal, once again mirroring the empirical results of \cite{schunk2009relationship}. Intuitively, due to the anticipation of experiencing regret over foregone candidates, biased decision-makers are likely to settle earlier than they should. In this respect, the theoretical model of \citet{kleinberg2021optimal} reflects human behavior relatively well, at least in the case of decisions involving scalar-valued candidates. 

However, their model does not capture decisions in which candidates have multiple features of interest. Notably, many key decisions in life including career, romantic partner, or house typically have many dimensions. Moreover, even simple consumer decisions such as selecting which movie to watch or magazine subscription to purchase often have multiple dimensions. Experimental research by    \citet{brenner1999comparison}  suggests that human decision makers suffer from ``comparative loss aversion'' in which they compare options along multiple dimensions, and perceive disadvantages as being more significant than advantages due to loss aversion, thus decreasing received utility from the option chosen. Consider the case of selecting a house: suppose that house A is the most aesthetically pleasing, house B is in the best location, and house C has the most space. Regardless of which house a person selects, due to comparative loss aversion, the mere awareness of the other options decreases the satisfaction they experience. Empirical research \cite{haynes2009testing, chernev2015choice,park2013confused} on the phenomenon of ``choice overload" (where people become less satisfied with their choice when they are presented with a larger set of options) further supports the existence of comparative loss aversion.

One interpretation of this phenomenon that has been popularized in psychological literature like \textit{The Paradox of Choice} \cite{schwartz2004paradox} is that people combine the best features of options they have previously seen into an imagined alternative of superior quality. Later empirical work by \citet{sagi2007cost} found that ``regret is
related to the comparison between the alternative chosen and the union of the positive attributes of the
alternatives rejected.'' Intuitively, compared to an imagined alternative combining the best attributes seen, the option chosen is much less satisfying. 

In our work, we model this phenomenon by associating each candidate $i$ with a \emph{$k$-dimensional vector} $\vec{\sigma}^{(i)} = \paren*{\sigma_1^{(i)}, \cdots, \sigma_k^{(i)}}$ rather than a single scalar. The agent's value for a single candidate is just the sum over the agent's value for the candidate on each dimension, namely, $\norm{\vec{\sigma}^{(i)}}_1$. So for an unbiased decision-maker, this is still just an instance of the classic prophet inequality. As in~\cite{kleinberg2021optimal}, our agent is biased, and maintains a reference point $\vec{s}^{(i)}$ at time $i$ based on options seen so far. The agent suffers from loss aversion parameterized by $\lambda$, and experiences (and anticipates) a loss equal to $\lambda \cdot \paren*{\norm{\vec{\sigma}^{(i)}}_1-\norm{\vec{s}^{(i)}}_1}$ upon selecting candidate $i$. The key distinction between our model and that of~\cite{kleinberg2021optimal} is in how the reference point is built. In~\citet{kleinberg2021optimal}, the reference point is simply the vector of a candidate $j \leq i$ maximizing the overall value $\norm{\sigma^{(j)}}_1$. In our model, the reference point is instead a hypothetical super-candidate that combines the best candidate seen so far for each feature. Formally, for each dimension $\ell$, the value of the reference point is $s^{(i)}_\ell = \max_{j \leq i}\{\sigma^j_\ell\}$. As previously discussed, this enables our model to better capture the multi-dimensional problem of purchasing a house, versus the single-dimensional problem of buying/selling stocks that are well-captured by the model of~\cite{kleinberg2021optimal}. To reiterate: our model is identical to that of~\cite{kleinberg2021optimal}, \emph{except} for the choice of reference point. Surprisingly, this single change from a single-dimensional reference point to a multi-dimensional reference point drastically changes conclusions in the model.

\subsection{Summary of Results}
Our main findings are that an agent who suffers from comparative loss aversion (i.e.~our model, the multi-dimensional case) suffers in a fundamentally different manner than an agent who only suffers from reference-dependent loss aversion (i.e.~\cite{kleinberg2021optimal}, the single-dimensional case). We now overview several results supporting this. In describing our results, ``gambler'' refers to the online decision-maker in the prophet inequality setting, and ``prophet'' refers to the offline decision-maker (who knows all options and simply picks the best).\\

\noindent\textbf{Unbounded Loss.} Perhaps our most surprising finding is a phase transition in $\lambda$ as a function of $k$ when comparing the utility of a biased gambler to that of an unbiased gambler or unbiased prophet. Specifically,~\citet{kleinberg2021optimal} show that a biased gambler can always guarantee at least a $\frac{1}{2 +\lambda}$-approximation to the expected utility achieved by an unbiased prophet, and a $\frac{1}{1 +\lambda}$-approximation to that of an unbiased gambler (and both of these bounds are tight). When $\lambda \cdot (k-1) < 1$, we find a natural generalization: the biased gambler can guarantee a $\frac{1-\lambda(k-1)}{2 +\lambda}$-approximation to the unbiased prophet, and a $\frac{1-\lambda(k-1)}{1 +\lambda}$-approximation to the unbiased gambler, and these are tight for all $k$ (including $k=1$, where it matches~\cite{kleinberg2021optimal}). However, there is a phase transition at $\lambda \cdot (k-1) = 1$: \emph{for any $\lambda \cdot (k-1) \geq 1$, and any $C$, there exist instances where the biased gambler achieves expected utility $1$, but the unbiased gambler (and unbiased prophet) achieve expected utility at least $C$}. That is, there is no non-trivial guarantee on the competitive ratio as a function only of $\lambda$ and $k$. Interestingly, this phase transition occurs for all $k > 1$, but not for $k=1$. This marks a fundamental difference between online decision-making with comparative loss aversion versus only reference-dependent loss aversion. 

Given the importance of the phase transition, we will call $\lambda \cdot (k-1)$ the \textit{\fbias}. Intuitively, \fbias{} measures how fast the gap between the utilities of unbiased and biased agents grows with the number of candidates.  As we will see in an illustrative example in section~\ref{example}, when the \fbias{} is greater than one, the ratio between utility of an unbiased prophet and that of a biased gambler grows exponentially with the number of candidates. The formal statements of our main results when \fbias{} is less than one can be found in theorems~\ref{thm:LBOfflineBound},~\ref{thm:LBOnlineBound} and \ref{thm:UBBoundedUtility}. For the setting where \fbias{} is at least one, the main results can be found in theorem~\ref{thm:advInftyGap}. \\

\noindent\textbf{No Improvement for Random Order, Optimal Order, and I.I.D. Distributions.} One key finding of~\citet{kleinberg2021optimal} is that the competitive ratio of a biased gambler compared to an unbiased gambler/prophet drastically improves (the loss is logarithmic in $\lambda$ instead of linear in $\lambda$) when the candidates are revealed in random order. This of course holds in our model when $k=1$ (because it is identical to~\cite{kleinberg2021optimal}). When $k > 1$, however, this phenomenon no longer occurs. Specifically, for all $k \geq 2$, \emph{even if the instance is i.i.d.}: (a) if \fbias{} $\lambda \cdot (k-1) \geq 1$, then the expected reward of the biased gambler and unbiased prophet/gambler can be unbounded, and (b) when \fbias{} $\lambda \cdot (k-1) < 1$, the worst-case ratio between the biased gambler and unbiased prophet/gambler is within a constant-factor of the worst-case ratio without the i.i.d.~assumption.\footnote{Of course, we do not nail down the tight competitive ratio for the gambler in the i.i.d. setting, because even the tight ratio between the gambler and prophet without loss aversion is quite involved. The tight competitive ratio between the unbiased gambler and unbiased prophet for the random order (but not i.i.d.) setting is still unknown.} The formal statements of our results for the i.i.d setting can be found in theorems~\ref{thm:iidInftyGap} and~\ref{thm:UBBoundedUtilityIID}.

\subsection{Related Work.}
Since its introduction, optimal stopping has been pertinent to the study of economic behaviors with immediate applications in a wide range of economic domains such as search theory, financial trading, auction design, etc. Various models have been proposed based on the optimal stopping paradigm including prophet inequalities (value distribution known), secretary problems (value distribution unknown) and Pandora's box problem (where a cost is associated with inspection of options).

The model that is most relevant to our present work is that of the prophet inequality, where the agents are given the distribution of their value for the candidates a priori. In the seminal work of \citet{KrengelS78}, it was shown that for prophet inequalities where the agent is limited to picking a single candidate, the ratio between the expected largest value (namely, the prophet's pick) and the expected value of the candidate an agent playing optimally selects online is at most $2$. Furthermore, this ratio is tight. \citet{Samuel-Cahn84} further show that by using a threshold based algorithm, the online agent can also guarantee that the expected value of the candidate they select is a $1/2$-approximation of the expected optimal value over all candidates. Prophet inequalities have since been further studied under a variety of feasibility constraints on the set of selected candidates (\citet{KleinbergW12, Rubinstein16, RubinsteinS17, DuettingFKL17}), and a variety of arrival order constraints (\citet{EsfandiariHLM15, EhsaniHKS18, Yan11, BeyhaghiGLPS18, HillK82, CorreaFHOV17}).

Although optimal stopping may have direct applications in human decision making, very little research has devoted to behavioral models of optimal stopping. Recently, \citet{kleinberg2021optimal} constructed a model which factors in reference dependence and loss aversion, however to our knowledge there is no other prior work exploring behavioral phenomena in the context of optimal stopping. 

\subsection{Organization.} 
The rest of the paper will be organized as follows. In section~\ref{model}, we discuss our model that incorporates comparative loss aversion and introduce notations that will be used throughout the paper. In section~\ref{example}, we present an illustrative example that highlights the difference between our model and the single dimensional model presented in \cite{kleinberg2021optimal}. In section~\ref{overview}, we analyze our model under various ordering constraints, with a focus on the gap between the utility of the biased \textit{gambler} and that of rational agents.
 In \Cref{futureWork} we discuss the limitations of our model and possible directions for future work. In \Cref{adversarial} and \Cref{relaxed}, we provide proofs that are omitted from \Cref{overview}. Lastly, in \Cref{app:monotonicity} we present monotonicity results which were all analogous to those of the one dimensional case presented in \citet{kleinberg2021optimal}. 
\section{Model and Preliminary} \label{model}
\subsection{Modeling Comparative Loss Aversion} 

\textbf{Multidimensional Optimal Stopping}. We now construct a multi-dimensional model of the optimal stopping problem such that candidates have multiple dimensions of value. First, we must formally define the multi-dimensional optimal stopping problem, using slight modifications of the classic definition of the problem in one dimension. Specifically, let there be $n$ candidates each of which has $k$ dimensions of value that are of interest to a decision making agent, i.e. there are $n$ vector valued candidates, such that their value vectors are $k$-dimensional. Note that the use of vector values differs from traditional constructions of the problem in which scalar values are used, i.e. in the traditional problem $k$ is 1. 

We will use $\mathcal{F}_{t}$ with support on $\mathbb{R}^k_+ $ to denote the prior distribution of the gambler's value vector for the $t^{th}$ candidate. Let $\sigma$ be an ordered sequence of $n$ candidates, where the $t^{th}$ candidate has a k-dimensional vector value $\sigma^{(t)} \in \mathbb{R}^k$ drawn from $\mathcal{F}_{t}$. Formally, we will use $\mathcal{F} = \times_{t \in [n]} \mathcal{F}_{t}$ to denote the joint distribution of the value for all candidates, such that $\sigma$ is drawn from $\mathcal{F}$.  Moreover, let it be the case that the agent sees the realized value of $\sigma^{(t)}$ only when they reach step t in the sequence. Upon reaching the $t^{th}$ candidate, $\sigma^{(t)}$, the agent must either accept them - thus halting their search - or decline them, continuing their search and losing the candidate forever. The agent attempts to select $\sigma^{(t)}$ to maximize their utility. We now need to reason about what an agent's utility function should look like. 
\newline 
\newline 
\textbf{The utility of agents with comparative loss aversion.} In keeping with the notation of \citet{kleinberg2021optimal}, let the parameter $\lambda \geq 0$ serve as a multiplier representing  our agent's loss aversion. In particular, let $\lambda$ be used to amplify the difference between the chosen candidate and the ``reference candidate'' along any dimensions for which the reference candidate is of higher value. Observe that in the case of our model, reasoning about the reference candidate is somewhat complex since an agent with comparative loss aversion compares each option against all previously seen along different dimensions independently.  Regardless, it is without loss of generality to collapse the highest values previously seen along each dimension into single vector and consider this to be the reference candidate. 

Accordingly, our agent remembers the best value previously seen along each dimension and combines these to form their reference candidate, which we refer to as the ``super candidate'', since is an upperbound on the quality of all previously seen candidates and combines all of their best features. Observe that this super candidate changes over time as we see new candidates that are better than all previously seen along some dimension. Due to the agent's reference dependence, when the super candidate is better than the agent's chosen candidate on any dimension, this detracts from the utility they receive. Let us denote the super candidate at step t in the sequence as $\vec{s}^{(t)}$ and the $j^{th}$ entry of the super candidate at step t as $\vec{s}^{(t)}_j$. We define the super candidate in step t element-wise as follows: 
$$ \vec{s}^{(t)}_j = \max\limits_{w : w \in [t]}( \sigma^{(w)}_j)$$
We define the utility received by a biased gambler, who suffers from loss aversion $\lambda$ and who selects the $t^{th}$ candidate in sequence $\sigma$, as follows: 
$$U_{g_b}(\sigma, t) = \sum_j \bigg( \sigma^{(t)}_j - \lambda(\vec{s}^{(t)}_j - \sigma^{(t)}_j)\bigg) = \|\sigma^{(t)}\|_{1} - \lambda \bigg(\|\vec{s}^{(t)}\|_1 - \|\sigma^{(t)}\|_1\bigg),$$
and we define the difference between the value selected by the biased gambler and their utility as the biased gambler's \textit{regret}.

Recall that, by construction, we have already restricted all values to be non-negative, allowing us to make use of L1 norms. Moreover, let's compare this to the utility received by biased prophet who selects option t:
$$U_{p_b}(\sigma, t) = \sum_j \bigg( \sigma^{(t)}_j - \lambda(\vec{s}^{(n)}_j - \sigma^{(t)}_j)\bigg) = \|\sigma^{(t)}\|_{1} - \lambda \bigg(\|\vec{s}^{(n)}\|_1 - \|\sigma^{(t)}\|_1\bigg)$$
Note that the only difference between the utility of the biased gambler and biased prophet is that the super option of the biased prophet is constructed from all n candidates, whereas that of the biased gambler is constructed only from the t candidates it has seen so far. 

Next let us reason about the utilities received by a rational gambler and rational prophet, respectively. In particular, observe that since they do not suffer from any reference dependence or loss aversion, their payoff is just equal to the value of the option they select:
$$U_{p_r}(\sigma, t) = U_{g_r}(\sigma, t) = \sum_j  \sigma^{(t)}_j = \| \sigma^{(t)} \|_1$$ 
Furthermore, assuming that our agents seek to maximize their utilities, the payoff of a rational prophet does not need parameter t since it selects t deterministically to maximize realized utility. Namely, it suffices to denote the utility of a rational prophet as follows: 
$$U_{p_r}(\sigma) = \max\limits_{t : t \in [n]}( U_{p_r}(\sigma, t))$$ 
Similarly, since the biased prophet sees all options and selects t deterministically to maximize realized utility, we can remove the parameter t here as well: 
$$U_{p_b}(\sigma) = \max\limits_{t : t \in [n]}( U_{p_b}(\sigma, t))$$

Now that the model has been developed, we can proceed to begin analyzing it. Next we will develop notations to describe the expected utility of biased and rational agents, and their relative differences. These notations will be used throughout the paper. 

\subsection{Preliminaries} 
In keeping with standard conventions, we use $\E_{\dist}$ to denote the expectation with respect to the distribution $\dist$. To abbreviate notation, we will use $\E_{\dist}\sq*{U^*_{g_b}}$ and $\E_{\dist}\sq*{U^*_{g_r}}$ to denote the expected utility of the biased gambler and a rational gambler, respectively, when they use a utility optimal online algorithm. Similarly, we will use $\E_{\dist}\sq*{U_{p_b}}$ and $\E_{\dist}\sq*{U_{p_r}}$ to denote the expected utility of the biased prophet and rational prophet, respectively. In addition, we will use $V^*$ to denote the maximum value of all candidates, namely, $V^* = \max \limits_{t \in [n]} \norm{\sigma^{(t)}}_1$. We will use $S_j^*$ to the denote the maximum value on the $j^{th}$ dimension, namely, $S_j^* = \max \limits_{t \in [n]} \sigma_j^{(t)}$.   

Next, we will create several definitions to describe the relationship between utility of the biased gambler and the utility of rational agents, which will be frequently referenced in subsequent sections. Informally, the prophet/online utility ratio describes how well the biased gambler can do compared to the rational prophet/gambler \textit{given a specific prior distribution $\dist$}. On the other hand, the prophet/online competitive ratio describes how well the biased gambler can do compared to the rational prophet/gambler \textit{given the most adversarial prior distribution $\dist$}. 

\begin{definition}[\textbf{Prophet Utility Ratio}]
    We will use the term \textit{prophet utility ratio}  to denote the ratio between the expected utility of a rational prophet and the expected optimal utility from a biased gambler under the prior distribution $\dist$. More formally, we define the prophet utility ratio as $\frac{\E_{\dist}\sq*{U_{p_r}}}{\E_{\dist}\sq*{U^*_{g_b}}}$.
\end{definition}
Note that the Prophet Utility Ratio captures the loss that the biased gambler experiences \emph{both due to bias and due to making decisions online}.

\begin{definition}[\textbf{Online Utility Ratio}]
    We will use the term \textit{online utility ratio} to denote the ratio between the expected optimal utility of a rational gambler and the expected optimal utility from a biased gambler under the prior distribution is $\dist$. Formally, we define the online utility ratio as  $\frac{\E_{\dist}\sq*{U_{g^*_r}}}{\E_{\dist}\sq*{U^*_{g_b}}}$.
\end{definition}
Note that the Online Utility Ratio captures the loss that the biased gambler experiences \emph{just due to bias}.

\begin{definition}[\textbf{Prophet/Online Competitive Ratio}]
    Given a setting (such as adversarial arrival order or i.i.d prior distribution), we define the terms \textit{prophet competitive ratio} and \textit{online competitive ratio} as the maximum prophet and online utility ratio possible in the setting. Specifically, in the adversarial arrival order setting, the prophet and online competitive ratio are formally defined as $\max \limits_{\dist} \frac{\E_{\dist}\sq*{U_{p_r}}}{\E_{\dist}\sq*{U^*_{g_b}}}$ and $\max \limits_{\dist} \frac{\E_{\dist}\sq*{U_{g^*_r}}}{\E_{\dist}\sq*{U^*_{g_b}}}$ , respectively. In the i.i.d prior distribution setting, we will denote $\mathcal{U}$ as the set of all i.i.d distributions. The prophet and online competitive ratio in the i.i.d prior distribution setting are then formally defined as 
$\max \limits_{\dist \in \mathcal{U}} \frac{\E_{\dist}\sq*{U_{p_r}}}{\E_{\dist}\sq*{U^*_{g_b}}}$ and $\max \limits_{\dist \in \mathcal{U}} \frac{\E_{\dist}\sq*{U_{g^*_r}}}{\E_{\dist}\sq*{U^*_{g_b}}}$.
\end{definition}
As stated in the introduction, we define $\lambda \cdot (k-1)$ to be the \textit{\fbias}. We will soon see that whether the \fbias{} $\lambda \cdot (k-1)$ is greater or less than one is critical in determining whether the biased gambler has a non trivial prophet/online competitive ratio.

Finally, the following claim that directly follows from classical prophet inequality results will be referenced and used throughout proofs in the rest of the paper.  
\begin{claim} \label{claim:classical}
    For any prior distribution $\dist$ and any assumption about candidate arrival order, $\E_{\dist}\sq*{U^*_{g_r}} \geq \frac{1}{2} \cdot \E_{\dist}\sq*{U_{p_r}}$. 
\end{claim}
\begin{proof}
    In our model, rational agents' utilities are simply their value for the option they select. Therefore our claim follows directly from the classical prophet inequality result \cite{KrengelS78, KleinbergW12} that the expected value from a value-optimal online algorithm is a $\frac{1}{2}$ approximation of the expected value from an offline prophet. 
\end{proof}

\section{Unbounded Loss: A Motivating Example} \label{example}
Prior to proving more general claims, we will first show an adversarial example that illustrates a fundamental difference between our multi-dimensional optimal stopping problem with biased agents and the single dimensional model. Specifically, in the single dimensional model, the biased gambler attain worse utility than rational agents only when there is \textit{uncertainty}. (Imagine if the biased agent already knows all values of the candidates, they can simply pick the one with the largest value and suffers no regret.) However, in our multi-dimensional model, we will show an example where the biased gambler attains arbitrarily worse utility than rational agents, even when the prior distribution is \textit{deterministic}. 

\begin{example} \label{example:motive}
Consider $n$ candidates arriving in adversarial order, where the biased gambler considers $k=2$ features of candidates and has loss aversion parameter $\lambda = 2$. Furthermore, let $\sigma$ be the following sequence of candidate values:
\begin{align*}
    (1, 0), (0, 1), (2, 0), (0, 2), (2^2, 0), (0, 2^2), \cdots, (2^{n/2 -1}, 0), (0, 2^{n/2 -1}),
\end{align*}
and let the prior distribution $\dist$ be deterministic -- where the realized sequence from $\dist$ is exactly $\sigma$ with probability one.  
\end{example}
In our example, the overall value of the candidates increases exponentially over time. Consequently, the rational prophet will receive high utility by simply selecting one of the last two candidates. However, the fact that the candidates are excellent on alternating dimensions causes the biased gambler to suffer from increasing regret as they see more candidates. In fact, should the biased gambler choose a later candidate rather than an earlier candidate, the increment in their regret roughly cancels out the increment in their value for the candidate.

With this intuition in mind, let's formally reason about the behaviors of our agents. Recall that the gambler is aware of the prior distribution $\dist$ and thus is aware of $\sigma$ a priori for this example.  The rational gambler is essentially the same as the rational prophet in this case since they know $\sigma$ a priori and experience no regret. Thus, the expected utility of both the rational prophet and the rational gambler is simply $2^{n/2 -1}$. 

Under our stylized model, the biased gambler only experiences regret for realized options, completely disregarding any future options they may be aware of a priori. Consequently, the biased gambler has utility $1$ if they just take the first candidate and are not presented other candidates, and they have utility $\leq 0$ if they take any other candidate.  (e.g. If the gambler decides to select let's say the candidate $(4, 0)$, then the super candidate has value $(4, 2)$, which means that the biased gambler's utility is $\norm{(4, 0)}_1 - 2 \cdot \paren*{\norm{(4, 2)}_1 - \norm{(4, 0)}_1} = 4 - 2 \cdot (6 - 4) = 0$.) Therefore we know $\E_{\dist}\sq*{U^*_{g_b}} \leq 1$. As a result, the biased gambler's prophet utility ratio, $\frac{\E_{\dist}\sq*{U_{p_r}}}{\E_{\dist} \sq*{U^*_{g_b}}} = \frac{U_{p_r}(\sigma)}{U_{g_b}(\sigma)}$, and online utility ratio, $\frac{\E_{\dist}\sq*{U^*_{g_r}}}{\E_{\dist} \sq*{U^*_{g_b}}} = \frac{U_{g_r}(\sigma)}{U_{g_b}(\sigma)}$, are both equal to $2^{n/2-1}$, and thus tend toward $\infty$ as $n$ increases. Recall that prophet and online utility ratios lower bound the prophet and online competitive ratio, respectively, implying that both competitive ratios must also be $\infty$.\\

As we will soon show in \Cref{claim:expUtilityGap} and \Cref{claim:linearUtilityGap}, our present adversarial example with candidates that are excellent on alternating dimensions can be generalized to any $k$ and $\lambda$ such that the \fbias{} $\lambda \cdot (k-1)$ is at least one in the adversarial setting. Furthermore, as we will show in \Cref{claim:iidExpUtilityGap} and \Cref{claim:iidlogUtilityGap}, we can extend our example to i.i.d setting (and thus to random and optimal candidate ordering setting) by carefully constructing the prior distribution so that with high probability, the sequence of realized candidate values presented to the agent is adversarial and (after deleting subsequent duplicates) similar to the sequence $\sigma$ in this example. Our capability to generalize the adversarial example under \textit{any} arrival order constraint drives our model's departure from the single dimensional model. Further, our adversarial example for the i.i.d setting is highly randomized, which shows that the adversarial examples are not an artifact of deterministic distributions.

Lastly, this example illustrates an additional property unique to our model: the biased gambler may outperform the biased prophet. Notice that in our example, the biased prophet is forced to see both $(2^{n/2-1}, 0)$ and $(0, 2^{n/2-1})$ as realized values, resulting in a reference point $(2^{n/2-1}, 2^{n/2-1})$ -- much larger than any value they could have chosen. Consequently, the biased prophet has negative utility and is worse off compared to the biased gambler. In contrast, in \cite{kleinberg2021optimal}'s model (where $k=1$), the biased prophet is guaranteed to attain at least as much utility as that of online agents. 

\section{Effects of Comparative Loss Aversion} \label{overview}
In this section we formally analyze the effect of comparative loss aversion on the biased gambler's utility. Specifically, we compare the biased gambler's utility with those of both the rational prophet and rational gambler.
We will focus on the two extremes of arrival order constraints: adversarial order and optimal order (even more specifically, we consider sequences drawn from an i.i.d. prior distribution, where any arrival order is optimal). Under both ordering constraints, we show a phase transition when the \fbias{} $\lambda \cdot (k-1)$ is equal to one. When the \fbias{} $\lambda (k-1) \geq 1$, the biased gambler's prophet and online competitive ratio are both equal to infinity. On the other hand, when the \fbias{} $\lambda (k-1) < 1$, the biased gambler can achieve a competitive ratio that is only dependent on $\lambda$ and $k$, and of similar magnitude under both ordering constraints. 

It is surprising that 
our model exhibits similar behavior in terms of competitive ratios regardless of constraints on arrival order. In contrast, \cite{kleinberg2021optimal} shows that for the single dimensional model, assuming the arrival order to be  random rather than adversarial results in an exponential improvement in utility guarantees in terms of $\lambda$. 

Missing proofs and the complete version of the proof sketches in \Cref{result:adversarial} and \Cref{result:IID} can be found in \Cref{adversarial} and \Cref{relaxed} respectively.

\subsection{Utility Loss under Adversarial Arrival Order} \label{result:adversarial}
Firstly, we generalize \Cref{example:motive} to show that when the \fbias{} $\lambda \cdot (k-1) \geq 1$, as $n$ grows arbitrarily large, the prophet and online utility ratios (namely, the ratio between expected utility from the biased gambler and the expected utility from the rational prophet or gambler) can be arbitrarily large. In particular, we construct a prior distribution $\dist_{0}(n)$, parameterized by the number of candidates $n$ such that, when the \fbias{} $\lambda \cdot (k-1) > 1$, both the prophet and online utility ratio are exponential in terms of $n$. When the \fbias{} $\lambda(k-1) = 1$, we obtain a more conservative result that prophet and online utility ratio can be linear in terms of $n$. 
\begin{restatable}{proposition}{claimExpUtilityGap} \label{claim:expUtilityGap} 
    When the \fbias{} $\lambda \cdot (k-1) > 1$, for any $n \in N^+$, there exists a sequence $\sigma$ with $n$ candidates such that when the distribution $\dist_0(n)$ takes the values of $\sigma$ deterministically, i.e. with probability 1, and when our prior distribution $\dist$ equals $\dist_0(n)$, the prophet utility ratio  $\frac{\E_{\dist}\sq*{U_{p_r}}}{\E_{\dist} \sq*{U^*_{g_b}}}$ and the online utility ratio $\frac{\E_{\dist}\sq*{U^*_{g_r}}}{\E_{\dist} \sq*{U^*_{g_b}}}$ are both equal to $\Omega\paren*{\paren*{\lambda (k-1)}^{\frac{n}{k} - 1}}$. 
\end{restatable} 
 
\begin{proof}[Proof Sketch]
    The proposition can be proven by considering the following deterministic distribution, where $\beta = \lambda \cdot (k-1)$: 
    \begin{center}
        \begin{tabular}{l l l c l}
        $(1, 0, \cdots, 0)$, &$(0, 1, 0, \cdots, 0)$, &$\cdots$, &$\cdots$, &$(0, \cdots, 0, 1)$,\\
        $(\beta, 0, \cdots, 0)$, &$(0, \beta, 0, \cdots, 0)$, &$\cdots$, &$\cdots$, &$(0, \cdots, 0, \beta)$,\\
        $(\beta^2, 0, \cdots, 0)$, &$(0, \beta^2, 0, \cdots, 0)$, &$\cdots$, &$\cdots$, &$(0, \cdots, 0, \beta^2)$,\\
        & &$\cdots$ & &\\
        $(\beta^{\ceil{n/k} - 1}, 0, \cdots, 0)$, &$(0, \beta^{\ceil{n/k} - 1}, 0, \cdots, 0)$, &$\cdots$, &$(0, \cdots, 0, \beta^{\ceil{n/k} - 1}, 0, \cdots, 0)$.&
        \end{tabular}
    \end{center}
    While the rational prophet and gambler are able to get utility $\beta^{\ceil{n/k} - 1}$ by selecting an option in the last row, the biased gambler could only get utility $1$, in particular by selecting the first option. 
\end{proof}

\begin{restatable}{proposition}{claimLinearUtilityGap} \label{claim:linearUtilityGap}
    When the \fbias{} $\lambda \cdot (k-1) = 1$, for any $n \in \mathbb{N}^+$, there exists a sequence $\sigma$ with $n$ candidates such that when the distribution $\dist_0(n)$ takes the values of $\sigma$ deterministically, i.e. with probability 1, and when our prior distribution $\dist$ equals $\dist_0(n)$, the prophet utility ratio  $\frac{\E_{\dist}\sq*{U_{p_r}}}{\E_{\dist} \sq*{U^*_{g_b}}}$ and the online utility ratio $\frac{\E_{\dist}\sq*{U^*_{g_r}}}{\E_{\dist} \sq*{U^*_{g_b}}}$ are both equal to $\Omega\paren*{\frac{n}{k}}$.
\end{restatable}
\begin{proof}[Proof Sketch]
    The proposition can be proven by considering the following deterministic distribution:
        \begin{center}
        \begin{tabular}{l l l c l}
        $(1, 0, \cdots, 0)$, &$(0, 1, 0, \cdots, 0)$, &$\cdots$, &$\cdots$, &$(0, \cdots, 0, 1)$,\\
        $(2, 0, \cdots, 0)$, &$(0, 2, 0, \cdots, 0)$, &$\cdots$, &$\cdots$, &$(0, \cdots, 0, 2)$,\\
        $(3, 0, \cdots, 0)$, &$(0, 3, 0, \cdots, 0)$, &$\cdots$, &$\cdots$, &$(0, \cdots, 0, 3)$,\\
        &&$\cdots$&& \\
        $(\ceil{\frac{n}{k}}, 0, \cdots, 0)$, &$(0, \ceil{\frac{n}{k}}, 0, \cdots, 0)$, &$\cdots$, &$(0, \cdots, 0, \ceil{\frac{n}{k}}, 0, \cdots, 0)$.
        \end{tabular}
    \end{center}
    Here, the rational prophet and gambler are able to get utility $\ceil{n/k}$, while the biased gambler could get utility no better than $1$. 
\end{proof}

\begin{restatable}{theorem}{thmadvInftyGap} \label{thm:advInftyGap}
    When the \fbias{} $\lambda \cdot (k-1) \geq 1$, the prophet and online competitive ratio are both equal to $\infty$.
\end{restatable}
\begin{proof}
    By \Cref{claim:expUtilityGap} and \Cref{claim:linearUtilityGap} we know that when $\lambda \cdot (k-1 ) \geq 1$, for any $n \in \mathbb{N}^+$, there exists a distribution $\dist_0(n)$ such that the prophet and online utility ratio are $\Omega(\frac{n}{k})$ (note that $\Omega\paren*{\paren*{\lambda (k-1)}^{\frac{n}{k} - 1}} = \omega(\frac{n}{k})$). Thus for any $T \in\mathbb{N}$, there exists a $n \in \mathbb{N}^+$ such that the prophet and online utility of $\dist_0(n)$ both exceed $T$. Since the prophet competitive ratio 
    $\max \limits_{\dist}  \frac{\E_{\dist}\sq*{U_{p_r}}}{\E_{\dist} \sq*{U^*_{g_b}}}$ is at least $\max \limits_{n \in \mathbb{N}^+} \frac{\E_{\dist_0(n)}\sq*{U_{p_r}}}{\E_{\dist_0(n)} \sq*{U^*_{g_b}}}$, any $T \in \mathbb{N}$ serves as a lower bound to the prophet utilty gap, and thus the prophet competitive ratio is equal to $\infty$. Similarly, the online competitive ratio is equal to $\infty$. 
\end{proof}

On the other hand, when the \fbias{} $\lambda(k-1) < 1$, we prove that the biased gambler, by using a simple threshold based algorithm, can get at least a $\frac{1 - \lambda(k-1)}{2 + \lambda}$ fraction of the prophet's utility and a $\frac{1 - \lambda(k-1)}{1 + \lambda}$ fraction of the optimal rational gambler's utility regardless of the value of $n$. We also show that both of these ratios are tight. More formally, we prove the following theorems. 

\begin{restatable}{theorem}{thmLBOfflineBound} \label{thm:LBOfflineBound}
    When the \fbias{} $\lambda \cdot (k-1) < 1$, for any prior distribution $\dist$, the inverse of the prophet utility ratio $\frac{\E_{\dist}\sq*{U^*_{g_b}}}{\E_{\dist}\sq*{U_{p_r}}} \geq (1 - \lambda (k-1))\cdot \max \set*{\frac{\gamma}{1 + \lambda + k}, \frac{1}{2 + \lambda}}$, where $\gamma = \frac{\E_{\dist}[\sum_{j=1}^k S_j^*]}{\E_{\dist}[V^*]}$. Consequently, the prophet competitive ratio is at most $\frac{2 + \lambda}{1 - \lambda (k-1)}$. 
\end{restatable}

\begin{restatable}{theorem}{thmLBOnlineBound} \label{thm:LBOnlineBound}
    When the \fbias{} $\lambda \cdot (k-1) < 1$, for any prior distribution $\dist$, the \oratio{} $\frac{\E_{\dist}\sq*{U^*_{g_r}}}{\E_{\dist}\sq*{U^*_{g_b}}} \leq \frac{1 + \lambda}{(1 - \lambda (k-1))}$. Consequently, the online competitive ratio is at most $\frac{1 + \lambda}{1 - \lambda (k-1)}$. 
\end{restatable}

\begin{restatable}{theorem}{thmUBBoundedUtility}  \label{thm:UBBoundedUtility}
    When the \fbias{} $\lambda \cdot (k-1) < 1$, for any $\delta \in (0, 1)$, there exists a prior distribution $\dist$ such that the prophet utility ratio $\frac{\E_{\dist} \sq*{U^*_{p_r}}}{\E_{\dist}\sq*{U_{g_b}}} \geq \frac{2 + \lambda}{1 - \lambda (k-1)} \cdot (1 - \delta)$ and that the online utility ratio $\frac{\E_{\dist} \sq*{U^*_{g_r}}}{\E_{\dist}\sq*{U^*_{g_b}}} \geq \frac{1 + \lambda}{1 - \lambda (k-1)} \cdot (1 - \delta)$. Consequently, the prophet competitive ratio is at least $\frac{2 + \lambda}{1 - \lambda (k-1)}$ and the online competitive ratio is at least $\frac{1 + \lambda}{1 - \lambda (k-1)}$. 
\end{restatable}

Our proof for \Cref{thm:LBOfflineBound} and \Cref{thm:LBOnlineBound} follows a similar outline to the single dimensional case presented in \citet{kleinberg2021optimal}. However, our analysis deals with the additional complexity from having multiple dimensions, and often requires a tighter analysis for intermediate terms in order to find the exact competitive ratios. For instance, in \Cref{thm:LBOfflineBound}, we require an additional term when bounding prophet utility ratio (which will later be used in \Cref{thm:LBOnlineBound}) compared to a similar claim in \cite{kleinberg2021optimal}. 

We will start by defining threshold based algorithms. 
\begin{definition}[Thresholding algorithms.]
    We define a thresholding algorithm parameterized by $\alpha \in (0, 1)$ and denoted $A^{\alpha}$ as follows. First, $A^{\alpha}$ sets threshold $T$ such that $\alpha = \Pr\sq*{V^* \geq T}$ (recall that $V^*$ is defined as the maximum value of all candidates). Next, during the selection process, the algorithm will select the first candidate whose value $\norm{\sigma^{(t)}}_1$ is at least the threshold $T$. We will denote the expected utility from $A^\alpha$ as $\E_{\dist}[U_{g_b}^\alpha]$.
\end{definition}

We will now reason about the biased gambler's utility from using a thresholding algorithm $A^\alpha$. In the following two claims, we first bound the utility (and regret) of a biased gambler that uses the thresholding algorithm when they encounter a particular sequence of values $\sigma$. This utility bound can then be used to bound the expected utility of the biased gambler given the prior distribution of values.

\begin{restatable}{claim}{claimUtilityBound} \label{claim:utilityBound}
    When the \fbias{} $\lambda \cdot (k-1) < 1$, for any sequence $\sigma$, for any threshold $T$, and for timestep $t$ such that $t$ is the smallest timestep where $\norm{\sigma^{(t)}}_1 \geq T$, it must be the case that $U_{g_b}(\sigma, t) \geq \norm{\sigma^{(t)}}_1 - \lambda (k-1) \cdot T$. 
\end{restatable}

\begin{restatable}{claim}{claimThresholdUtility} \label{claim:thresholdUtility}
    When the \fbias{} $\lambda \cdot (k-1) < 1$, for any prior distribution $\dist$ and any $\alpha \in (0, 1)$, $\E_{\dist}\sq*{U_{g_b}^\alpha} \geq  \paren*{(1 + \lambda)\alpha -k \lambda} \cdot T + (1 - \alpha) \cdot \sum_{t \in [n]} \E_{\sigma^{(t)} \sim \dist_t} \sq*{(\norm{\sigma^{(t)}}_1 - T)^+}$. 
\end{restatable}

\begin{proof}[Proof Sketch]
Let $T$ be the threshold where $\alpha = \Pr\sq*{V^* \geq T}$, then $T$ is the threshold used by the algorithm $A^\alpha$. Observe that when the gambler uses the algorithm $A^\alpha$, with probability $\alpha$ the algorithm will select some item, and with probability $1 - \alpha$ the algorithm will select nothing. 

Conditioned on $A^\alpha$ selecting no item, we know that all items have value $< T$. Thus even at the end of the algorithm, the gambler's value for the super candidate must not exceed $k \cdot T$. This means that the gambler's utility is at least $-\lambda \cdot k \cdot T$.  On the other hand, we can prove by using \Cref{claim:utilityBound} that conditioned on $A^\alpha$ selecting an item at step $t$, the gambler's utility is at least $(1 - \lambda(k-1) )\cdot T +  \frac{\E_{\sigma^{(t)} \sim \dist_t} \sq*{(\norm{\sigma^{(t)}}_1 - T)^+}}{\Pr\sq*{\norm{\sigma^{(t)}}_1 \geq T}}$. The proof of the claim then follows a standard argument for prophet inequality where we take the expectation of agent utility over all events (for all t: selecting an item at step $t$; selecting no item).
\end{proof}

We can now bound the utility ratio between a biased gambler using algorithm $A^\alpha$ and the rational prophet, which can then be used to prove \Cref{thm:LBOfflineBound} and \Cref{thm:LBOnlineBound}. It turns out that \Cref{thm:LBOfflineBound} and \Cref{thm:LBOnlineBound} requires two different ways of bounding the the utility ratio.

\begin{restatable}{lemma}{lemLBOfflineBound} \label{lem:LBOfflineBound}
    When the \fbias{} $\lambda \cdot (k-1) < 1$, for any prior distribution $\dist$, there exists an $\alpha$ such that $\frac{\E_{\dist}\sq*{U_{g_b}^\alpha}}{\E_{\dist}\sq*{U_{p_r}}} \geq (1 - \lambda (k-1))\cdot \max \set*{\frac{\gamma}{1 + \lambda + k}, \frac{1}{2 + \lambda}}$, where $\gamma = \frac{\E_{\dist}[\sum_{j=1}^k S_j^*]}{\E_{\dist}[V^*]}$ (recall that $S_j^*$ denotes the maximum value in the sequence along dimension $j$ and $V^* = \max \limits_t \norm{\sigma^{(t)}}_1$). 
\end{restatable}
\begin{proof}[Proof Sketch]
The key step in this lemma is to bound the term $\displaystyle \sum_{t \in [n]} \E_{\sigma^{(t)} \sim \dist_t} \sq*{(\norm{\sigma^{(t)}}_1 - T)^+}$ from \Cref{claim:thresholdUtility}. 

Let us call $(\norm{\sigma^{(t)}}_1 - T)^+$ the surplus from the $t^{th}$ candidate. Note that $\displaystyle \sum_{t \in [n]} \E_{\norm{\sigma^{(t)}}_1 \sim \dist_t} \sq*{(\norm{\sigma^{(t)}}_1 - T)^+}$ is the sum over the expected surplus from each candidate. Thus, it must be at least the expected surplus from the candidate with the largest overall value. Moreover, it must be at least the sum over the expected surplus from candidates that have the largest value on each individual dimension. Formally, the following two inequalities always hold:  
    \begin{align} \label{eq:sumToMax}
        \sum_{t \in [n]} \E_{\sigma^{(t)} \sim \dist_t} \sq*{(\norm{\sigma^{(t)}}_1 - T)^+} \geq \E_{\dist}\sq*{(V^* - T)^+} \geq \E_{\dist}\sq*{V^* - T},
    \end{align}
    and 
    \begin{align} \label{eq:sumToMaxDim}
        \sum_{t \in [n]} \E_{\sigma^{(t)} \sim \dist_t} \sq*{(\norm{\sigma^{(t)}}_1 - T)^+} \geq \sum_{j=1}^k \E_{\dist}\sq*{(S_j^* - T)^+} \geq  \E_{\dist}\sq*{\sum_{j=1}^k S_j^*} - k\cdot T \geq \gamma \cdot \E_{\dist}\sq*{V^*} - k \cdot T. 
    \end{align}
\Cref{eq:sumToMax} and \Cref{eq:sumToMaxDim} can then be used to bound the utility of the biased gambler, and in particular its relation to the utility of the rational prophet. 
\end{proof}

We are now ready to prove \Cref{thm:LBOfflineBound} and \Cref{thm:LBOnlineBound}. 

\begin{proof}{(of \Cref{thm:LBOfflineBound})}
        The first line of the theorem follows directly by \Cref{lem:LBOfflineBound} and the fact that $\E_{\dist}\sq*{U^*_{g_b}} \geq \E_{\dist}\sq*{U_{g_b}^\alpha}$. (Because $\E_{\dist}\sq*{U^*_{g_b}}$ is defined as the expected utility from a biased gambler that uses the utility optimal online algorithm,  $\E_{\dist}\sq*{U^*_{g_b}}$ is at least the expected utility from a biased gambler using some specific online algorithm). 
    
    Now we will derive the second line of the theorem. From the first line of the theorem we know that for any prior distribution $\dist$, 
    \begin{align*}
        \frac{\E_{\dist}\sq*{U^*_{g_b}}}{\E_{\dist}\sq*{U_{p_r}}} \geq (1 - \lambda (k-1))\cdot \max \set*{\frac{\gamma}{1 + \lambda + k}, \frac{1}{2 + \lambda}} \geq \frac{1 - \lambda (k-1))}{2 + \lambda}. 
    \end{align*}
    Hence, for any prior distribution $\dist$, the prophet utilit ratio $\frac{\E_{\dist}\sq*{U_{p_r}}}{\E_{\dist}\sq*{U^*_{g_b}}}$ is at most $ \frac{2 + \lambda}{1 - \lambda (k-1))}$. We conclude that the prophet competitive ratio $\max \limits_{\dist} \frac{\E_{\dist}\sq*{U_{p_r}}}{\E_{\dist}\sq*{U^*_{g_b}}}$ is at most $ \frac{2 + \lambda}{1 - \lambda (k-1))}$. 
\end{proof}

\begin{proof}{(of \Cref{thm:LBOnlineBound})}
    Let $\beta = \frac{\E_{\dist}\sq*{U_{p_r}}}{\E_{\dist} \sq*{U^*_{g_b}}}$ and let $\Delta = \frac{\E_{\dist}\sq*{U_{p_r}}}{\E_{\dist} \sq*{U^*_{g_r}}}$. Since $\frac{\E_{\dist}\sq*{U^*_{g_r}}}{\E_{\dist}\sq*{U^*_{g_b}}} = \frac{\beta}{\Delta}$, proving our theorem is equivalent to proving that $\frac{\Delta}{\beta} \geq \frac{(1 - \lambda (k-1))}{1 + \lambda}$, and we will do exactly this for the rest of the proof. 

    Firstly, let $\gamma = \frac{\E_{\dist}[\sum_{j=1}^k S_j^*]}{\E_{\dist}[V^*]}$, then by \Cref{thm:LBOfflineBound} we know that $\frac{1}{\beta} = \frac{\E_{\dist}\sq*{U^*_{g_b}}}{\E_{\dist}\sq*{U_{p_r}}} \geq (1 - \lambda (k-1)) \cdot \frac{\gamma}{1 + \lambda + k}$. 
    
    Next, because $\E_{\dist}\sq*{U^*_{g_b}}$ represents the highest expected utility a biased gambler can gain from using an online algorithm, $\E_{\dist}\sq*{U^*_{g_b}}$ must be at least the expected utility a biased gambler can gain from using the value-optimal online algorithm, which is equal to the optimal utility of a rational gambler minus the loss aversion the biased gambler experience from using the value-optimal online algorithm. Recall that $S_j^*$ is defined as the maximum value on the $j^{th}$ dimension among all candidates, thus the value of super candidate the biased gambler experiences loss aversion against must be at most $\sum_{j = 1}^k S_j^*$. Thus 
    \begin{align*}
        \E_{\dist}\sq*{U^*_{g_b}} \geq \E_{\dist} \sq*{U^*_{g_r} - \lambda \cdot \paren*{\sum_{j=1}^k S_j^* - U^*_{g_r}}} = (1 + \lambda) \cdot \E_{\dist} \sq*{U^*_{g_r}} - \lambda \cdot \E_{\dist}\sq*{\sum_{j=1}^k S_j^*}. 
    \end{align*}
    By the definition of $\Delta$, we know that $\E_{\dist} \sq*{U^*_{g_r}} = \frac{1}{\Delta} \cdot \E_{\dist}\sq*{U_{p_r}}$. Similarly, by the definition of $\gamma$,  $\E_{\dist}\sq*{\sum_{j=1}^k S_j^*} = \gamma \cdot \E_{\dist}[V^*] = \gamma \cdot \E_{\dist}\sq*{U_{p_r}}$. Therefore combining these two statements with the previous inequality we know the following: 
    \begin{align*}
        \E_{\dist}\sq*{U^*_{g_b}} \geq \paren*{\frac{1 + \lambda}{\Delta} - \lambda \cdot \gamma} \cdot \E_{\dist}\sq*{U_{p_r}}. 
    \end{align*}
    Re-expressing algebraically, we have $\frac{1}{\beta} = \frac{\E_{\dist}\sq*{U^*_{g_b}}}{\E_{\dist}\sq*{U_{p_r}}} \geq \paren*{\frac{1 + \lambda}{\Delta} - \lambda \cdot \gamma} = \frac{1 + \lambda -\lambda \cdot \gamma \cdot \Delta}{\Delta}$. 

    Now we have two bounds on $\frac{1}{\beta}$: we know that $\frac{1}{\beta} \geq (1 - \lambda (k-1)) \cdot \frac{\gamma}{1 + \lambda + k}$ and that $\frac{1}{\beta} \geq \frac{1 + \lambda -\lambda \cdot \gamma \cdot \Delta}{\Delta}$. By these two bounds, we conclude that 
    \begin{align*}
        \frac{\Delta}{\beta}  = \Delta \cdot \frac{1}{\beta} \geq \max \set*{(1 - \lambda (k-1)) \cdot \frac{\Delta \cdot \gamma}{1 + \lambda + k}, \enspace 1 + \lambda -\lambda \cdot \gamma \cdot \Delta}. 
    \end{align*}
    Notice that inside the max bracket, the first term monotonically increases with $\Delta$, while the second term monotonically decreases with $\Delta$, therefore the minimum of the max of two terms occurs when the two terms are equal, namely, when $(1 - \lambda (k-1)) \cdot \frac{\Delta \cdot \gamma}{1 + \lambda + k} = 1 + \lambda -\lambda \cdot \gamma \cdot \Delta$. This is achieved when $\Delta = \frac{1 + \lambda + k}{\gamma \cdot (1 + \lambda)}$, at which point
    \begin{align*}
        (1 - \lambda (k-1)) \cdot \frac{\Delta \cdot \gamma}{1 + \lambda + k} = 1 + \lambda -\lambda \cdot \gamma \cdot \Delta = \frac{1 - \lambda(k-1)}{1 + \lambda}. 
    \end{align*}
    We thus conclude that it is always the case that $\frac{\E_{\dist}\sq*{U^*_{g_b}}}{\E_{\dist}\sq*{U^*_{g_r}}} = \frac{\Delta}{\beta} \geq \frac{1 - \lambda(k-1)}{1 + \lambda}$ regardless of the value of $\Delta$. 
\end{proof}

We remain to discuss  \Cref{thm:UBBoundedUtility}, which we prove by constructing a worst case prior distribution. Our construction combines a generalization of \Cref{example:motive} with the canonical two item example showing a $1/2$-competitive ratio for classical prophet inequality. In our construction, all but the last candidate has deterministic values constructed similarly to those in \Cref{example:motive} (but the value of the candidates grow much slower due to the fact that the \fbias{} is less than one). The last candidate has a high value with a very small probability, which is similar to constructions for classical prophet inequality. 
\begin{proof}[Proof Sketch]{(of \Cref{thm:UBBoundedUtility})} Let the first $n-1$ candidates have the following deterministic value: 
    \begin{center}
        \begin{tabular}{ l l l l }
            $(1, 0, \cdots, 0)$, &$(0, 1, 0, \cdots, 0)$, &$\cdots$, &$(0, \cdots, 0, 1)$,\\
            $(1 + \beta, 0, \cdots, 0)$, &$(0, 1 + \beta, 0, \cdots, 0)$, &$\cdots$, &$(0, \cdots, 0, 1 + \beta)$,\\
            $(1 + \beta + \beta^2, 0, \cdots, 0)$, &$(0, 1 + \beta + \beta^2, 0, \cdots, 0)$, &$\cdots$, &$(0, \cdots, 0, 1 + \beta + \beta^2)$,\\
            &$\cdots$ & &\\
            $(\sum_{j = 0}^{\ceil{n/k} - 1} \beta^{j}, 0 , \cdots, 0)$, &$(0, \sum_{j = 0}^{\ceil{n/k} - 1} \beta^{j}, 0, \cdots, 0)$, &$\cdots$, &$(0, \cdots, 0, \sum_{j = 0}^{\ceil{n/k} - 1} \beta^{j}).$  
        \end{tabular}
    \end{center}
    Then, the last randomized candidate is equal to $\paren*{(1  - \epsilon)(1 + \lambda) \cdot \sum_{j = 0}^{\ceil{n/k} - 1} \beta^{j}/\epsilon, 0, \cdots, 0}$ with probability $\epsilon$ and is equal to $(0, 0, \cdots, 0)$ with probability $1 - \epsilon$. One can then verify that the biased gambler can only get utility at most $1$, since by the time the biased gambler reach the last randomized candidate, their imaginary super candidate is too strong for the biased gambler to gain positive expected utility from the randomized candidate. However, the rational prophet can simply select the maximum between the best deterministic option and the randomized option, getting utility that approaches $\frac{2+\lambda}{1 - \lambda(k-1)}$ as $\epsilon$ approaches $0$ and $n$ approaches $\infty$. Similarly, the rational gambler can simply always select the randomized option and get utility that approaches and $\frac{1+\lambda}{1-\lambda(k-1)}$. 
\end{proof}

\subsection{No Improvement Under Relaxed Arrival Order} \label{result:IID}
We now show that most of our results for adversarial arrival order can be used to derive analogous results for the case where the prior distribution is i.i.d. 
 
Specifically, in the next claim (\Cref{claim:representation}), we will show that removing repeated candidate values from a sequence does not affect the biased gambler's optimal utility. We will then use the claim to establish in a subsequent lemma (\Cref{lem:detToiid}) that for each deterministic distribution (i.e. each distribution such that values drawn from it are equal to some sequence $\sigma$ with probability $1$), we can find an i.i.d prior distribution $\dist$ with a similar realized sequence (after deleting duplicates), and hence similar prophet and online utility ratio. As we have discussed in \Cref{example} and \Cref{result:adversarial}, the construction of the worst case instances for the adversarial arrival order setting is either a fully deterministic sequence, or with few randomized options. Consequently, \Cref{lem:detToiid} can be leveraged to translate all lower bounds on the utility ratios in adversarial setting into lower bounds in the i.i.d setting. 

\begin{definition}
    A sequence $\sigma$ is called \textit{succinct} if there are no two identical candidates in the sequence and no candidate's value is a zero vector. 
\end{definition}
\begin{definition}
    The \textit{representation} $r(\sigma)$ of a sequence $\sigma$ is the sequence of unique elements in $\sigma$ ordered in the  ascending order of their first occurance in $\sigma$. (e.g. r(1, 2, 1, 2, 3) = 1, 2, 3.)
\end{definition}

\begin{claim} \label{claim:representation}
    $U_{p_r}(\sigma) = U_{p_r}(r(\sigma))$ and $U_{g_b}(\sigma) = U_{g_b}(r(\sigma))$.   
\end{claim}
\begin{proof}
    It is clear that $U_{p_r}(\sigma) = U_{p_r}(r(\sigma))$ since $r(\sigma)$ contains all the unique elements in $\sigma$. Let $a$ be an arbitrary candidate in $\sigma$, let $I_a$ be the set of all indices $t$ where $\sigma^{(t)} = a$ and let $t_a$ be the minimum index in $I_a$. Let $\sigma'$ be the same sequence as $\sigma$ except for $\norm{\sigma^{(t)}}_1 = 0$ for all $t \in I_a \setminus \{t_a\}$. We will show that $U_{g_b}(\sigma') = U_{g_b}(\sigma)$. 
    
    Firstly, for any index $t$, the super candidate $s^{(t)}$ is only affected by the set of unique candidates at indices $< t$. Therefore $s^{(t)}$ takes the same value in $\sigma$ and $\sigma'$ for any $t \not \in I_a \setminus \{t_a\}$. As a result, $U_{g_b}(\sigma, t) = U_{g_b}(\sigma', t)$ for any $t \not \in I_a \setminus \{t_a\}$. 
    
    Also, when seeing the same candidate value, having seen more candidates in the past can only decrease the biased gambler's utility. Thus for any $t \in I_a$, $U_{g_b}(\sigma, t_a) \geq U_{g_b}(\sigma, t)$. We conclude that $U_{g_b}(\sigma) = \max_{t \in [n]} U_{g_b}(\sigma, t)= \max_{t \in [n]}(\sigma', t) = U_{g_b}(\sigma')$. 

    By applying the above construction repeatedly, we can arrive at a sequence $\sigma^*$ where the only non zero entries are the first occurrence of each candidate in $\sigma$. $r(\sigma)$ is simply the result of removing candidates that have value $0$ from $\sigma^*$. Clearly, this operation does not affect the biased gambler's maximum utility. Thus $U_{g_b}(r(\sigma)) = U_{g_b}(\sigma^*) = U_{g_b}(\sigma)$.  

\end{proof}

\begin{lemma} \label{lem:detToiid}
    For any $\epsilon \in (0, 1)$, any succinct sequence $\sigma$ such that its length $m = |\sigma|$ is greater than $1$, there exists a distribution $\dist$ where all candidates are i.i.d with support on the set $\{\sigma^{(1)}, \cdots, \sigma^{(m)}\}$ such that the prophet utility ratio $\frac{\E_{\dist}\sq*{U_{p_r}}} {\E_{\dist} \sq*{U^*_{g_b}}} \geq (1 - \epsilon) \cdot \frac{U_{p_r}(\sigma)}{U_{g_b}(\sigma)}$.
\end{lemma}

\begin{proof}
    Given an $x \in (0, 1)$ that we will set later, consider a distribution $\mathcal{D}$ that takes value $\sigma_i$ ($i^{th}$ candidate in the sequence $\sigma$) with probability $\frac{x^{i-1}(1 - x)}{1 - x^m}$ for all $i \in [m]$. First we will verify that $\mathcal{D}$ is a valid distribution (namely, the probability sums up to be 1): 
    \begin{align*}
        \sum_{i=1}^{m} \frac{x^{i-1}(1 - x)}{1 - x^m} = \frac{1 - x}{1 - x^m} \sum_{i=1}^m x^{i-1} = \frac{1 - x^m}{1 - x} \cdot \frac{1-x}{1 - x^m} = 1. 
    \end{align*} 
    Now, let $\dist = \times_{i \in [n]}\mathcal{D}$, where $n = m^{\alpha (m-1)} \cdot \log^{\alpha}(m)$ for some $\alpha > 1$ that we will set later. Let $\wsigma$ be a sequence drawn from $\dist$, from \Cref{claim:representation} we know that as long as $r(\wsigma) = \sigma$, then $U_{p_r}(\wsigma) = U_{p_r}(\sigma)$ and $U_{g_b}(\wsigma) = U_{g_b}(\sigma)$. 
    Furthermore, by construction the maximum value in a sequence $\wsigma$ that is drawn from $\dist$ is at most the maximum value in $\sigma$, consequently,  $U_{g_b}(\wsigma) \leq U_{p_r}(\wsigma) \leq U_{p_r}(\sigma)$. 
    Therefore 
    \begin{align*}
        \frac{\E_{\dist} \sq*{U^*_{g_b}}}{\E_{\dist}\sq*{U_{p_r}}} &\leq \frac{\Pr_{\wsigma \sim \dist} \sq*{r(\wsigma) = \sigma} \cdot U_{g_b}(\sigma) + \Pr_{\wsigma \sim \dist} \sq*{r(\wsigma) \neq \sigma}\cdot U_{p_r}(\sigma)}{\Pr_{\wsigma \sim \dist} \sq*{r(\wsigma) = \sigma} \cdot U_{p_r}(\sigma)} \\
        & \leq \frac{U_{g_b}(\sigma)}{U_{p_r}(\sigma)} + \frac{\Pr_{\wsigma \sim \dist} \sq*{r(\wsigma) \neq \sigma}}{\Pr_{\wsigma \sim \dist} \sq*{r(\wsigma) = \sigma}}.
    \end{align*} 
    Now we want to bound the probability that $r(\wsigma) \neq \sigma$. There are two ways for this to happen, either $\wsigma$ does not contain all candidates in $\sigma$, or $r(\wsigma)$ is ordered differently than $\sigma$. Formally, by union bound, 
    \begin{align*}
        \Pr_{\wsigma \sim \dist} \sq*{r(\wsigma) \neq \sigma} \leq  \sum_{i \in [m]} \Pr_{\wsigma \sim \dist} \sq*{\sigma_i \not \in r(\wsigma)} + \sum_{i \in [m-1]} \Pr_{\wsigma \sim \dist} \cond{\sigma_{i+1} \text{ occur before } \sigma_{i} \text{ in } r(\wsigma)}{\sigma_i, \sigma_{i+1} \in r(\wsigma)}. 
    \end{align*}
    Firstly, we bound the probability that $\wsigma$ does not contain $\sigma_i$ as follows 
    \begin{align*}
        \Pr_{\wsigma \sim \dist} \sq*{\sigma_i \not \in r(\wsigma)} = \paren*{1 - \frac{x^{i-1} (1-x)}{1 - x^m}}^n \leq \paren*{1 - \frac{x^{m-1} (1-x)}{1 - x^m}}^n \leq \paren*{1 - x^{m-1}}^n. 
    \end{align*}
    Next, let $r_{i,i+1}(\wsigma)$ denote the smallest index $j$ such that $\wsigma_{j} = \sigma_i$ or $\sigma_{i+1}$. 
    \begin{align*}
        \Pr_{\wsigma \sim \dist} \cond{\sigma_{i+1} \text{ occur before } \sigma_{i} \text{ in } r(\wsigma)}{\sigma_i, \sigma_{i+1} \in r(\wsigma)} &= \Pr_{\wsigma \sim \dist} \cond{\wsigma_{r_{i,i+1}(\wsigma)} = \sigma_{i+1}}{\sigma_i, \sigma_{i+1} \in r(\wsigma)} \\
        &= \Pr_{v \sim \mathcal{D}}\cond{v = \sigma_{i+1}}{v \in \{\sigma_{i+1}, \sigma_{i}\}} \\
        &= \frac{\frac{x^{i-1}(1 - x)}{1 - x^m}}{\frac{x^{i-1}(1 - x)}{1 - x^m} + \frac{x^{i}(1 - x)}{1 - x^m}} = \frac{x}{x+1}.
    \end{align*}
    Now we set $x = m^{-\alpha}$, and conclude that 
    \begin{align*}
        \Pr_{\wsigma \sim \dist} \sq*{r(\wsigma) \neq \sigma} &\leq m \cdot \paren*{1 - x^{m-1}}^n + (m-1) \cdot \frac{x}{x+1}\\
        &\leq m \cdot \paren*{1 - m^{-\alpha \cdot (m-1)}}^{m^{\alpha (m-1)} \cdot \log^{\alpha}(m)} +  (m-1) \cdot m^{-\alpha} \\
        & \leq m \cdot m^{-\alpha} + (m - 1) \cdot m ^{-\alpha}\\
        & \leq m^{- (\alpha - 1)} + m^{- (\alpha - 1)}. 
    \end{align*}
    By choosing $\alpha = \log_{m} \paren*{\frac{U_{p_r}(\sigma)}{\epsilon \cdot U_{g_b}(\sigma)}} + 2$, the above entity is $< \epsilon \cdot \frac{U_{g_b}(\sigma)}{U_{p_r}(\sigma)}$, thus 
    \begin{align*}
        \frac{\E_{\dist} \sq*{U^*_{g_b}}}{\E_{\dist}\sq*{U_{p_r}}} &\leq \frac{U_{g_b}(\sigma)}{U_{p_r}(\sigma)} + \frac{\Pr_{\wsigma \sim \dist} \sq*{r(\wsigma) \neq \sigma}}{\Pr_{\wsigma \sim \dist} \sq*{r(\wsigma) = \sigma}} \leq \frac{U_{g_b}(\sigma)}{U_{p_r}(\sigma)} + \frac{\epsilon \cdot \frac{U_{g_b}(\sigma)}{U_{p_r}(\sigma)}}{1 - \epsilon \cdot \frac{U_{g_b}(\sigma)}{U_{p_r}(\sigma)}}.
    \end{align*} 
    Since $\frac{U_{g_b}(\sigma)}{U_{p_r}(\sigma)}$ is at most $1$, we conclude that $\frac{\E_{\dist} \sq*{U^*_{g_b}}}{\E_{\dist}\sq*{U_{p_r}}} \leq \paren*{1 + \frac{\epsilon}{1 - \epsilon}} \cdot \frac{U_{g_b}(\sigma)}{U_{p_r}(\sigma)}$. Consequently, $$\frac{\E_{\dist}\sq*{U_{p_r}}}{\E_{\dist} \sq*{U^*_{g_b}}} \geq (1 - \epsilon) \cdot \frac{U_{p_r}(\sigma)}{U_{g_b}(\sigma)}.$$
\end{proof}

We can now use \Cref{lem:detToiid} to prove that when the candidate values are i.i.d, it is also the case that when the \fbias{} $\lambda (k-1) \geq 1$, prophet and online utility ratios can get progressively worse as the number of candidates $n$ increases, and as a result, both the prophet and online competitive ratio are equal to $\infty$.

\begin{restatable}{proposition}{claimIIDExpUtilityGap} \label{claim:iidExpUtilityGap}
When the \fbias{} $\lambda \cdot (k-1) > 1$, for all $n \in \mathbb{N}^+$, there exists a prior distribution $\dist_0(n)$ where all candidates are i.i.d such that when $\dist = \dist_0(n)$, the \pratio{} $\frac{\E_{\dist}\sq*{U_{p_r}}}{\E_{\dist}\sq*{U^*_{g_b}}}$ and the \oratio{} $\frac{\E_{\dist}\sq*{U^*_{g_r}}} {\E_{\dist}\sq*{U^*_{g_b}}}$ are both equal to  
$\Omega \paren*{(\lambda(k-1))^{f(\lambda, k) \cdot \log^{\frac{1}{2}} n }}$, where $f(\lambda, k)$ is a function that only depends on $\lambda$ and $k$. 
\end{restatable}
\begin{proof}[Proof Sketch]
    Let $\sigma(m)$ be the adversarial sequence considered in \Cref{claim:expUtilityGap} of length $m$. Then $\frac{U_{p_r}(\sigma(m))}{U_{g_b}(\sigma(m))} = \Omega\paren*{(\lambda(k-1))^{\frac{m}{k}-1}}$. By construction $\sigma(m)$ contains distinct entries, and thus is succinct. Thus by \Cref{lem:detToiid}, we can construct a randomized and i.i.d distribution $\dist$ with $n = m^{\alpha (m-1)} \cdot \log^{\alpha}(m)$ candidates where $\alpha = \log_{m}\paren*{\frac{(\lambda(k-1))^{\frac{m}{k}-1}}{\epsilon }} + 2$ that has a similar prophet and online utility ratio (within $\epsilon$ multiplicative difference). A careful relation between $n$ and $m$ then proves the claim. 
\end{proof}

\begin{restatable}{proposition}{claimIIDlogUtilityGap} \label{claim:iidlogUtilityGap}
    When the \fbias{} $\lambda \cdot (k-1) = 1$, for all $n \in \mathbb{N}^+$, there exists a prior distribution $\dist_0(n)$ where all candidates are i.i.d such that when $\dist = \dist_0(n)$, the \pratio{} $\frac{\E_{\dist}\sq*{U_{p_r}}}{\E_{\dist}\sq*{U^*_{g_b}}}$ and the \oratio{} $\frac{\E_{\dist}\sq*{U^*_{g_r}}}{\E_{\dist}\sq*{U^*_{g_b}}}$ are both equal to $\Omega\paren*{
        \frac{\log n}{k \cdot \log \log n}
    }$.
\end{restatable}
\begin{proof}[Proof Sketch]
   The proof follows similarly to \Cref{claim:iidlogUtilityGap}, except that we consider the adversarial sequence $\sigma(m)$ from \Cref{claim:linearUtilityGap}, which has a prophet and online utility ratio of $\Omega(\frac{m}{k})$. 
\end{proof}

\begin{restatable}{theorem}{thmIIDInftyGap} \label{thm:iidInftyGap}
    In the i.i.d prior distribution setting, when the \fbias{} $\lambda \cdot (k-1) \geq 1$, the prophet and online competitive ratio are both equal to $\infty$.  
\end{restatable}

On the other hand, when the \fbias{} $\lambda (k-1) < 1$, we can utilize \Cref{lem:detToiid} to show a lower bound for the prophet and online competitive ratios that matches those in the adversarial setting within constant factors. 
\begin{restatable}{claim}{claimUBBoundedUtilityIID} \label{claim:UBBoundedUtilityIID}
    When the \fbias{} $\lambda (k-1) <1$, for any $\delta \in (0, 1)$, there exists an i.i.d prior distribution $\dist$ where the \pratio{} $\frac{\E_{\dist}\sq*{U_{p_r}}}{\E_{\dist}\sq*{U^*_{g_b}}} \geq (1 - \delta) \cdot \frac{1}{(1 - \lambda (k-1))}$. 
\end{restatable}
\begin{proof}[Proof Sketch] 
    We will consider a sequence of length $m$ that is similar to the mostly deterministic distribution considered in theorem~\ref{thm:UBBoundedUtility}, except that the last candidate, which has randomized value, is excluded. Specifically, let $\sigma$ be the following sequence:
    \begin{center}
        \begin{tabular}{ l l l l }
            $(1, 0, \cdots, 0)$, &$(0, 1, 0, \cdots, 0)$, &$\cdots$, &$(0, \cdots, 0, 1)$,\\
            $(1 + \beta, 0, \cdots, 0)$, &$(0, 1 + \beta, 0, \cdot(1 - \lambda (k-1))(1 + \delta)s, 0)$, &$\cdots$, &$(0, \cdots, 0, 1 + \beta)$,\\
            $(1 + \beta + \beta^2, 0, \cdots, 0)$, &$(0, 1 + \beta + \beta^2, 0, \cdots, 0)$, &$\cdots$, &$(0, \cdots, 0, 1 + \beta + \beta^2)$,\\
            &$\cdots$ & &\\
            $(\sum_{i = 0}^{\ceil{\frac{m}{k}} - 1} \beta^{i}, 0 , \cdots, 0)$, &$(0, \sum_{i = 0}^{\ceil{\frac{m}{k}} - 1} \beta^{i}, 0, \cdots, 0)$, &$\cdots$, &$(0, \cdots, 0, \sum_{i = 0}^{\ceil{\frac{m}{k}} - 1} \beta^{i})$.  
        \end{tabular}
    \end{center}   
    For a deterministic distribution equal to $\sigma$, the prophet utility ratio is $\frac{1}{1 - \lambda(k-1)}$. 
    We will then apply \Cref{lem:detToiid} to obtain an i.i.d distribution with a larger number of candidates that has a similar prophet utility ratio.  
\end{proof}
\begin{restatable}{theorem} {thmUBBoundedUtilityIID} \label{thm:UBBoundedUtilityIID}
    When the \fbias{} $\lambda (k-1) <1$ and $k > 1$, for any $\delta \in (0, 1)$, there exists an i.i.d prior distribution $\dist$ where the \pratio{} $\frac{\E_{\dist}\sq*{U_{p_r}}}{\E_{\dist}\sq*{U^*_{g_b}}} \geq  \frac{(1 - \delta)}{3} \cdot \frac{2 + \lambda}{1 - \lambda (k-1)}$ and the \oratio{} $\frac{\E_{\dist}\sq*{U^*_{g_r}}}{\E_{\dist}\sq*{U^*_{g_b}}} \geq \frac{(1 - \delta)}{4} \cdot \frac{1 + \lambda}{1 - \lambda (k-1)}$. Consequently, the prophet competitive ratio is at least $\frac{1}{3} \cdot \frac{2 + \lambda}{(1 - \lambda (k-1))}$ and the online competitive ratio is at least $\frac{1}{4} \cdot \frac{1 + \lambda}{(1 - \lambda (k-1))}$.
\end{restatable}

\begin{proof}
    By \Cref{claim:UBBoundedUtilityIID}, there exists an i.i.d prior distribution $\dist$ such that $\frac{\E_{\dist}\sq*{U_{p_r}}}{\E_{\dist}\sq*{U^*_{g_b}}} \geq (1 - \delta) \cdot \frac{1}{(1 - \lambda (k-1))}$. Combining this with \Cref{claim:classical}, we have  $\frac{\E_{\dist}\sq*{U^*_{g_r}}}{\E_{\dist}\sq*{U^*_{g_b}}} \geq \frac{(1 - \delta)}{2} \cdot (1 - \lambda (k-1))$. 
    
    Since $\lambda (k-1) < 1$ and $k > 1$, it must be the case that $2 + \lambda < 3$ and that $1 + \lambda < 2$. Hence $\frac{\E_{\dist}\sq*{U_{p_r}}}{\E_{\dist}\sq*{U^*_{g_b}}} \geq \frac{(1 - \delta)}{3} \cdot \frac{2 + \lambda}{1 - \lambda (k-1)}$ and $\frac{\E_{\dist}\sq*{U^*_{g_r}}}{\E_{\dist}\sq*{U^*_{g_b}}} \geq \frac{(1 - \delta)}{4} \cdot \frac{1 + \lambda}{1 - \lambda (k-1)}$. This implies that the upper bound matches the lower bounds proved for adversarial setting within constant factors. 
\end{proof}

\section{Conclusion and Directions for Future Work} \label{futureWork}

We consider an online decision-maker who suffers from comparative loss aversion. Our model is identical to that of~\cite{kleinberg2021optimal}, \emph{except} for how the reference point is formed. This change alone accounts for fundamental differences in conclusions. For example, even the prophet suffers loss in our model, and this enables novel phenomena to occur. Additionally, the loss due to bias in our model can be unboundedly bad, and is not significantly mitigated by assumptions on the prophet inequality arrival order. 

Empirically observed biases have only recently been incorporated into research on optimal stopping. Our work takes a step beyond~\cite{kleinberg2021optimal} and uncovers new phenomena by changing just one aspect of their model. An important direction for future work is to continue developing models that uncover additional phenomena of behavioral game theory (one such possibility is those studied in~\cite{simonson1992choice}, where adding inferior options can influence decision-making).

\bibliographystyle{plainnat}
\bibliography{bibliography}
\appendix
\section{Omitted Proofs of Section \ref{result:adversarial}} \label{adversarial}

\claimExpUtilityGap*
\begin{proof}
    Let $\beta = \lambda \cdot (k-1)$, consider a sequence $\sigma$ with $n$ candidates where the $t^{th}$ candidate takes value $\sigma^{(t)}$ such that for dimension $j = (t \mod k)$, $\sigma^{(t)}_{j} = \beta^{\ceil{\frac{t}{k}} - 1}$ and for any dimension $j \neq (t \mod k)$, $\sigma^{(t)}_{j} = 0$. Below is a more intuitive representation of the sequence $\sigma$: 
    \begin{center}
        \begin{tabular}{l l l c l}
        $(1, 0, \cdots, 0)$, &$(0, 1, 0, \cdots, 0)$, &$\cdots$, &$\cdots$, &$(0, \cdots, 0, 1)$,\\
        $(\beta, 0, \cdots, 0)$, &$(0, \beta, 0, \cdots, 0)$, &$\cdots$, &$\cdots$, &$(0, \cdots, 0, \beta)$,\\
        $(\beta^2, 0, \cdots, 0)$, &$(0, \beta^2, 0, \cdots, 0)$, &$\cdots$, &$\cdots$, &$(0, \cdots, 0, \beta^2)$,\\
        & &$\cdots$ & &\\
        $(\beta^{\ceil{n/k} - 1}, 0, \cdots, 0)$, &$(0, \beta^{\ceil{n/k} - 1}, 0, \cdots, 0)$, &$\cdots$, &$(0, \cdots, 0, \beta^{\ceil{n/k} - 1}, 0, \cdots, 0)$.&
        \end{tabular}
    \end{center}
    Clearly, if the biased gambler selects the first candidate, they will get utility $1$. Further, we will show that the biased gambler gets utility at most $0$ whenever it accepts a candidate that is not the first candidate.

    Firstly, observe that for any fixed $i$, among all the candidates in $i^{th}$ row, the biased gambler's utility for the first candidate in $i^{th}$ row is the highest. This is because the gambler has the same value for all candidates in $i^{th}$ row, however, the super candidate's value increases as the gambler sees more items, and thus the utility of the gambler decreases as they see more items in the same row. 

    Thus it suffices for us to consider the biased gambler's utility when they accept the first candidate in the $i^{th}$ row of the above matrix for some $i > 0$. In this case, the super candidate $s^{((i-1)\cdot k + 1)} = (\beta^{i-1}, \beta^{i-2}, \cdots, \beta^{i-2})$. We conclude that 
    \begin{align*}
        U_{g_b}(\sigma, (i-1)\cdot k + 1) &= \norm{\sigma^{((i-1)\cdot k + 1)}}_1 - \lambda \cdot \paren*{\norm{s^{((i-1)\cdot k + 1)}}_1 - \norm{\sigma^{((i-1) \cdot k + 1)}}_1}\\
        &= \beta^{i-1} - \lambda \cdot (k-1) \cdot \beta^{i-2} = \beta^{i-1} - \beta \cdot \beta^{i-2} = 0.
    \end{align*}

    Thus, under $\sigma$, the optimal strategy for the biased gambler is to accept the very first candidate, yielding $\E_{\dist} \sq*{U^*_{g_b}} = 1$.
    Meanwhile, the rational prophet's expected utility $\E_{\dist}\sq*{U_{p_r}}$ is equal to the expected largest value, which is $\beta^{\ceil{\frac{n}{k}} - 1} \geq \paren*{\lambda (k-1)}^{\frac{n}{k} - 1}$. Thus $\frac{\E_{\dist}\sq*{U_{p_r}}}{\E_{\dist} \sq*{U^*_{g_b}}} = \Omega\paren*{\paren*{\lambda (k-1)}^{\frac{n}{k} - 1}}$. 
    
    By claim~\ref{claim:classical} this immediately implies that $\frac{\E_{\dist}\sq*{U^*_{g_r}}}{\E_{\dist} \sq*{U^*_{g_b}}} = \Omega\paren*{\paren*{\lambda (k-1)}^{\frac{n}{k} - 1}}$ as well. 
\end{proof}

\claimLinearUtilityGap*
\begin{proof}
    We will consider a sequence $\sigma$ with $n$ candidates where the $t^{th}$ candidate takes value $\sigma^{(t)}$ such that for dimension $j = (t \mod k)$, $\sigma^{(t)}_{j} = \ceil{\frac{t}{k}}$ and for any dimension $j \neq (t \mod k)$, $\sigma^{(t)}_{j} = 0$. Below is a more intuitive representation of the sequence $\sigma$:
    
    We will consider an example similar to the one in claim~\ref{claim:expUtilityGap}, where $\sigma$ is equal to the following sequence:
    \begin{center}
        \begin{tabular}{l l l c l}
        $(1, 0, \cdots, 0)$, &$(0, 1, 0, \cdots, 0)$, &$\cdots$, &$\cdots$, &$(0, \cdots, 0, 1)$,\\
        $(2, 0, \cdots, 0)$, &$(0, 2, 0, \cdots, 0)$, &$\cdots$, &$\cdots$, &$(0, \cdots, 0, 2)$,\\
        $(3, 0, \cdots, 0)$, &$(0, 3, 0, \cdots, 0)$, &$\cdots$, &$\cdots$, &$(0, \cdots, 0, 3)$,\\
        &&$\cdots$&& \\
        $(\ceil{\frac{n}{k}}, 0, \cdots, 0)$, &$(0, \ceil{\frac{n}{k}}, 0, \cdots, 0)$, &$\cdots$, &$(0, \cdots, 0, \ceil{\frac{n}{k}}, 0, \cdots, 0)$.
        \end{tabular}
    \end{center}
    We will show that the biased gambler gets utility at most $1$ no matter what candidate they accept.

    For the same reason as in claim~\ref{claim:expUtilityGap} it suffice for us to consider the biased gambler's utility when they accept the first candidate in the $i^{th}$ row of the above matrix. In this case, the super candidate $s^{((i-1)\cdot k + 1)} = (i, i-1, \cdots, i-1)$. Thus
    \begin{align*}
        U_{g_b}(\sigma, (i-1)\cdot k + 1) &= \norm{\sigma^{((i-1)\cdot k + 1)}}_1 - \lambda \cdot \paren*{\norm{s^{((i-1)\cdot k + 1)}}_1 - \norm{\sigma^{((i-1) \cdot k + 1)}}_1}\\
        &= i - \lambda \cdot (k-1) \cdot (i-1) = i - (i-1) = 1.
    \end{align*}

    Thus, under $\sigma$, one of the optimal strategies for the biased gambler is to accept the very first candidate, yielding $\E_{\dist} \sq*{U^*_{g_b}} = U_{g_b}(\sigma, 1) = 1$. On the other hand, the rational prophet's expected utility $\E_{\dist}\sq*{U_{p_r}}$ is equal to the expected largest value, which is at least $\frac{n}{k}$. Thus $\frac{\E_{\dist}\sq*{U_{p_r}}}{\E_{\dist} \sq*{U^*_{g_b}}} = \Omega(\frac{n}{k})$. By claim~\ref{claim:classical}, this immediately implies that $\frac{\E_{\dist}\sq*{U^*_{g_r}}}{\E_{\dist} \sq*{U^*_{g_b}}} = \Omega(\frac{n}{k})$ as well. 
\end{proof}

\claimUtilityBound*
\begin{proof}
    We know by definition that
    \begin{align*}
        U_{g_b}(\sigma, t) &= \norm{\sigma^{(t)}}_1 - \lambda \cdot (\norm{\vec{s}^{(t)}}_1 - \norm{\sigma^{(t)}}_1) \\
        &= \norm{\sigma^{(t)}}_1 - \lambda \cdot (\norm{\vec{s}^{(t)}}_1 - \norm{\sigma^{(t)}}_1) = (1 + \lambda)\cdot \norm{\sigma^{(t)}}_1 - \lambda \cdot \norm{\vec{s}^{(t)}}_1.
    \end{align*}

    Since $t$ is the smallest timestep where $\norm{\sigma^{(t)}}_1 \geq T$, we know that  $\norm{\sigma^{(t')}}_1 < T$ for any $t' < t$. Thus, the value at any particular index of the super candidate at step $t-1$ is at most $T$. As a result, for any $j$, $\norm{\vec{s}^{(t)}_j}_1 \leq \max(\norm{\sigma^{(t)}_j}_1, \norm{\vec{s}^{(t-1)}_j}_1) \leq \max(\norm{\sigma^{(t)}_j}_1, T)$. This implies that 
    \begin{align*}
        \norm{\vec{s}^{(t)}}_1 = \sum_j \vec{s}^{(t)}_j \leq \sum_j \max(\sigma^{(t)}_j, T) &= \sum_j T + \sum_{j} \paren*{\sigma^{(t)}_j - T}^+.  
    \end{align*}
    Observe that since $\norm{\sigma^{(t)}}_1 - T \geq 0$, in the case where for all $j$, $\norm{\sigma^{(t)}}_1 - T < 0$,  $\sum_{j} \paren*{\sigma^{(t)}_j - T}^+ = 0 \leq \norm{\sigma^{(t)}}_1 \geq T$. On the other hand, when for at least one dimension $j^*$, $\norm{\sigma^{(t)}}_1 - T \geq 0$, we can write $\sum_{j} \paren*{\sigma^{(t)}_j - T}^+ = \sigma^{(t)}_{j^*} - T + \sum_{j \neq j^*} \paren*{\sigma^{(t)}_j - T}^+$, which is at most $\sigma^{(t)}_{j^*} - T + \sum_{j \neq j^*} \sigma^{(t)}_j = \norm{\sigma^{(t)}}_1 - T$. We thus conclude that 
    \begin{align*}
        \norm{\vec{s}^{(t)}}_1 \leq k \cdot T + (\norm{\sigma^{(t)}}_1 - T) = \norm{\sigma^{(t)}}_1 + (k-1) \cdot T.
    \end{align*}
    Combining everything together, we get that 
    \begin{align*}
        U_{g_b}(\sigma, t) = (1 + \lambda)\cdot \norm{\sigma^{(t)}}_1 - \lambda \cdot \norm{\vec{s}^{(t)}}_1 &\geq (1 + \lambda)\cdot \norm{\sigma^{(t)}}_1 - \lambda \cdot ( \norm{\sigma^{(t)}}_1 + (k-1) \cdot T) \\
        &\geq \norm{\sigma^{(t)}}_1 - \lambda (k-1) \cdot T. 
    \end{align*}
\end{proof}

\claimThresholdUtility*

\begin{proof}
    Let $T$ be the threshold where $\alpha = \Pr\sq*{V^* \geq T}$, then $T$ is the threshold used by the algorithm $A^\alpha$. Observe that when the gambler uses the algorithm $A^\alpha$, with probability $\alpha$ the algorithm will select some item, and with probability $1 - \alpha$ the algorithm will select nothing. 

    Conditioned on $A^\alpha$ selecting no item, we know that all items have value $< T$. Thus even at the end of the algorithm, the gambler's value for the super candidate must not exceed $k \cdot T$. This means that the gambler's utility is at least $-\lambda \cdot k \cdot T$.  

    On the other hand, we will prove by using claim~\ref{claim:utilityBound} that conditioned on $A^\alpha$ selecting an item at step $t$, the gambler's utility is at least $\norm{\sigma^{(t)}}_1 - \lambda(k-1) \cdot T$. We will begin by showing a block of formal deduction, which we subsequently unpack in the following paragraphs.
    \begin{align*}
        \E_{\dist} \cond*{U_{g_b}^\alpha}{A^\alpha \text{ selects item }t} &= \E_{\sigma \sim \dist} \cond*{U_{g_b}(\sigma, t)}{t = \argmin_{i \in [n]} \{\norm{\sigma^{(i)}}_1 \geq T\} } \\
        &\geq \E_{\sigma \sim \dist} \cond*{\norm{\sigma^{(t)}}_1 - \lambda(k-1) \cdot T}{t = \argmin_{i \in [n]} \{\norm{\sigma^{(i)}}_1 \geq T\} }\\
        &\geq (1 - \lambda(k-1) )\cdot T +  \E_{\sigma \sim \dist} \cond*{(\norm{\sigma^{(t)}}_1 - T)^+}{\norm{\sigma^{(t)}}_1 \geq T}\\
        &\geq (1 - \lambda(k-1) )\cdot T +  \E_{\sigma^{(t)} \sim \dist_t} \cond*{(\norm{\sigma^{(t)}}_1 - T)^+}{\norm{\sigma^{(t)}}_1 \geq T}\\
        &\geq (1 - \lambda(k-1) )\cdot T +  \frac{\E_{\sigma^{(t)} \sim \dist_t} \sq*{(\norm{\sigma^{(t)}}_1 - T)^+}}{\Pr\sq*{\norm{\sigma^{(t)}}_1 \geq T}}.
    \end{align*}
    In the above argument, the first line is due to $A^\alpha$ selecting the first candidate with value at least $T$. Furthermore, the second line follows from the first line due to claim~\ref{claim:utilityBound}. The second line to the third line is because the value of the candidates are independent, and thus ``$\sigma^{(t)}$ is the first candidate that has value at least $T$'' and ``$\sigma^{(t)}$ has value at least $T$'' have equivalent implications for the value of $\sigma^{(t)}$. The fourth line follows from the third because our formula only involves $\sigma^{(t)}$ and thus the expectation over $\sigma$ drawn from $\dist$ is equivalent to the expectation over $\sigma^{(t)}$ drawn from $\dist_t$.  The forth line to the fifth line is because when $\norm{\sigma^{(t)}}_1 < T$, $(\norm{\sigma^{(t)}}_1 - T)^+$ takes the value $0$, thus $$\E_{\sigma^{(t)} \sim \dist_t} \sq*{(\norm{\sigma^{(t)}}_1 - T)^+} = \E_{\sigma^{(t)} \sim \dist_t} \cond*{(\norm{\sigma^{(t)}}_1 - T)^+}{\norm{\sigma^{(t)}}_1 \geq T} \cdot \Pr\sq*{\norm{\sigma^{(t)}}_1 \geq T}.$$ 

    We will now bound the expected utility of the biased gambler when they use the algorithm $A^\alpha$. The argument is presented in the inequalities immediately below, which will be unpacked in the following paragraphs. 
    \begin{align*}
        \E_{\dist} \sq*{U_{g_b}^\alpha} &= \Pr\sq*{U_{g_b}^\alpha \text{ selects no item}}  \cdot (-k \lambda \cdot T) + \sum_{t \in [n]} \Pr\sq*{U_{g_b}^\alpha \text{ selects item }j} \cdot \E \cond*{U_{g_b}^\alpha}{A^\alpha \text{ selects item }j}  \\
        &\geq (1 - \alpha) \cdot (-k \lambda \cdot T) +   \sum_{t \in [n]} \Pr\sq*{U_{g_b}^\alpha \text{ selects item }t} \cdot \paren*{(1 - \lambda(k-1))\cdot T +   \frac{\E_{\sigma^{(t)} \sim \dist_t} \sq*{(\norm{\sigma^{(t)}}_1 - T)^+}}{\Pr\sq*{\norm{\sigma^{(t)}}_1 \geq T}}}\\
        &\geq (1 - \alpha) \cdot (-k \lambda \cdot T) + \alpha \cdot (1 - \lambda(k-1))\cdot T \\
        &\quad \quad \quad \quad + (1 - \alpha) \cdot \sum_{t \in [n]} \Pr \sq*{\norm{\sigma^{(t)}}_1 \geq T} \cdot \frac{\E_{\sigma^{(t)}_1 \sim \dist_t} \sq*{(\norm{\sigma^{(t)}} - T)^+}}{\Pr\sq*{\norm{\sigma^{(t)}}_1 \geq T}}\\
        &\geq \paren*{(1 + \lambda)\alpha -k \lambda} \cdot T + (1 - \alpha) \cdot \sum_{t \in [n]} \E_{\sigma^{(t)} \sim \dist_t} \sq*{(\norm{\sigma^{(t)}}_1 - T)^+}.
    \end{align*}
    For the above equation, the first and second line is just summing over all possible conditions and using inequalities that we established previous. The second line to the third line is because $\sum_{t \in [n]} \Pr\sq*{U_{g_b}^\alpha \text{ selects item }j} = \alpha$ and the following inequalities,  
    \begin{align*}
        \Pr\sq*{U_{g_b}^\alpha \text{ selects item }t} &= \Pr\sq*{\forall i < t,  \norm{\sigma^{(i)}}_1 < T} \cdot \Pr\sq*{\norm{\sigma^{(t)}}_1 \geq T} \\
        &\geq \Pr\sq*{\forall i \in [n],  \norm{\sigma^{(i)}}_1 < T} \cdot \Pr\sq*{\norm{\sigma^{(t)}}_1 \geq T} = (1 - \alpha) \cdot \Pr\sq*{\norm{\sigma^{(t)}}_1 \geq T}.
    \end{align*}
    The third line to the fourth line is simply due to a cancellation of terms. 
\end{proof}

\lemLBOfflineBound*
\begin{proof}
        Let $\alpha \in (0,1)$ be a fixed value that we will set later. By claim~\ref{claim:thresholdUtility}, we know that 
    \begin{align*}
        \E_{\dist}\sq*{U_{g_b}^\alpha} \geq  \paren*{(1 + \lambda)\alpha -k \lambda} \cdot T + (1 - \alpha) \cdot \sum_{t \in [n]} \E_{\sigma^{(t)} \sim \dist_t} \sq*{(\norm{\sigma^{(t)}}_1 - T)^+}. 
    \end{align*}

    In order to further reason about the above quantity, we need to bound $\displaystyle \sum_{t \in [n]} \E_{\sigma^{(t)} \sim \dist_t} \sq*{(\norm{\sigma^{(t)}}_1 - T)^+}$, which is what we will do next. 
    
    Let us call $(\norm{\sigma^{(t)}}_1 - T)^+$ the surplus from the $t^{th}$ candidate. Note that $\displaystyle \sum_{t \in [n]} \E_{\norm{\sigma^{(t)}}_1 \sim \dist_t} \sq*{(\norm{\sigma^{(t)}}_1 - T)^+}$ is the sum over the expected surplus from each candidate. Thus, it must be at least the expected surplus from the candidate with the largest overall value. Moreover, it must be at least the sum over the expected surplus from candidates that have the largest value on each individual dimension. Formally, the following two inequalities always hold:  
    \begin{align} \label{eq:sumToMax2}
        \sum_{t \in [n]} \E_{\sigma^{(t)} \sim \dist_t} \sq*{(\norm{\sigma^{(t)}}_1 - T)^+} \geq \E_{\dist}\sq*{(V^* - T)^+} \geq \E_{\dist}\sq*{V^* - T},
    \end{align}
    and 
    \begin{align} \label{eq:sumToMaxDim2}
        \sum_{t \in [n]} \E_{\sigma^{(t)} \sim \dist_t} \sq*{(\norm{\sigma^{(t)}}_1 - T)^+} \geq \sum_{j=1}^k \E_{\dist}\sq*{(S_j^* - T)^+} \geq  \E_{\dist}\sq*{\sum_{j=1}^k S_j^*} - k\cdot T \geq \gamma \cdot \E_{\dist}\sq*{V^*} - k \cdot T. 
    \end{align}
    Now, we will use these bounds to derive the final bound on the biased gambler's utility from using algorithm $A^\alpha$. By combining inequality \ref{eq:sumToMax2} and claim~\ref{claim:thresholdUtility}, we know that 
    \begin{align*}
        \E_{\dist} \sq*{U_{g_b}^\alpha} \geq \paren*{(1 + \lambda)\alpha -k \lambda} \cdot T + (1 - \alpha) \cdot \E_{\dist}\sq*{V^* - T}. 
    \end{align*}
    When we set $\alpha = \frac{\lambda k + 1}{2 + \lambda}$, $\paren*{(1 + \lambda)\alpha -k \lambda} = 1 - \alpha = \frac{1 - \lambda(k - 1)}{2 + \lambda}$. This means that for this choice of $\alpha$, the preceding inequality becomes: 
    \begin{align*}
        \E_{\dist} \sq*{U_{g_b}^\alpha} \geq \frac{1 - \lambda(k - 1)}{2 + \lambda} \cdot (T + \E\sq*{V^* - T}) \geq \frac{1 - \lambda(k - 1)}{2 + \lambda} \cdot \E_{\dist}\sq*{V^*} = \frac{1 - \lambda(k - 1)}{2 + \lambda} \cdot \E_{\dist}\sq*{U_{p_r}}. 
    \end{align*}
    
    On the other hand, by combining inequality \ref{eq:sumToMaxDim2} and claim~\ref{claim:thresholdUtility}, we know that 
    \begin{align*}
        \E_{\dist} \sq*{U_{g_b}^\alpha} &\geq \paren*{(1 + \lambda)\alpha -k \lambda} \cdot T + (1 - \alpha) \cdot \paren*{\gamma \cdot \E_{\dist}\sq*{V^*} - k \cdot T}\\
        &\geq  \paren*{(1 + \lambda)\alpha - k \cdot \lambda - (1 - \alpha) \cdot k} \cdot T + (1 - \alpha) \cdot \paren*{\gamma \cdot \E_{\dist}\sq*{V^*}}. 
    \end{align*}
    When we set $\alpha = \frac{k \cdot (1 + \lambda)}{1 + \lambda + k} $, $\paren*{(1 + \lambda)\alpha - k \cdot \lambda - (1 - \alpha) \cdot k} = 0$, therefore for this choice of $\alpha$,
    \begin{align*}
        \E \sq*{U_{g_b}^\alpha} \geq (1 - \alpha) \cdot \gamma \cdot \E_{\dist}\sq*{V^*} = \paren*{1 - \frac{k \cdot (1 + \lambda)}{1 + \lambda + k}} \cdot \gamma \cdot \E_{\dist}\sq*{V^*} = (1 - \lambda(k-1)) \cdot \frac{\gamma}{1 + \lambda + k} \cdot \E_{\dist}\sq*{U_{p_r}}. 
    \end{align*}
    Combining the above two cases we conclude that 
    \begin{align*}
        \E \sq*{U_{g_b}^\alpha} \geq \max \set*{\frac{1}{2 + \lambda}, \frac{\gamma}{1 + \lambda + k}} \cdot (1 - \lambda(k-1)) \cdot \E_{\dist}\sq*{U_{p_r}}, 
    \end{align*}
    which implies our lemma statement. 
\end{proof}
 
% lower bound example 
\thmUBBoundedUtility*
\begin{proof}
    Let $\beta = \lambda \cdot (k-1)$, and fix some positive integer $w \in \mathbb{N}^+$ and some small constant $\epsilon$ that we will set later. In our example there will be $n = w \cdot k + 1$ candidates, and all candidates will be deterministic except for the last candidate, i.e. $\sigma^{(1)}....\sigma^{(n-1)}$ are drawn from deterministic $\dist_1 ....\dist_{n-1}$ but $\sigma^{n}$ is drawn from stochastic $\dist_n$. We define the sequence of deterministic values as follows: 
    \begin{center}
        \begin{tabular}{ l l l l }
            $(1, 0, \cdots, 0)$, &$(0, 1, 0, \cdots, 0)$, &$\cdots$, &$(0, \cdots, 0, 1)$,\\
            $(1 + \beta, 0, \cdots, 0)$, &$(0, 1 + \beta, 0, \cdots, 0)$, &$\cdots$, &$(0, \cdots, 0, 1 + \beta)$,\\
            $(1 + \beta + \beta^2, 0, \cdots, 0)$, &$(0, 1 + \beta + \beta^2, 0, \cdots, 0)$, &$\cdots$, &$(0, \cdots, 0, 1 + \beta + \beta^2)$,\\
            &$\cdots$ & &\\
            $(\sum_{j = 0}^{w - 1} \beta^{j}, 0 , \cdots, 0)$, &$(0, \sum_{j = 0}^{w - 1} \beta^{j}, 0, \cdots, 0)$, &$\cdots$, &$(0, \cdots, 0, \sum_{j = 0}^{w - 1} \beta^{j}).$  
        \end{tabular}
    \end{center}

    The last randomized candidate is equal to $\paren*{(1  - \epsilon)(1 + \lambda) \cdot \sum_{j = 0}^{w - 1} \beta^{j}/\epsilon, 0, \cdots, 0}$ with probability $\epsilon$ and is equal to $(0, 0, \cdots, 0)$ with probability $1 - \epsilon$. We will prove that given our prior distribution $\dist$, the gambler's expected utility is always at most $1$ no matter which candidate they choose. 

    Notice that as with similar constructions in claim~\ref{claim:expUtilityGap} and claim~\ref{claim:linearUtilityGap}, if the gambler picks a deterministic candidate in some row $i$, then their utility is maximized by picking the first candidate in row $i$, which results in utility
    \begin{align*}
        \E_{\sigma \sim \dist}\sq*{U_{g_b}(\sigma, i \cdot k + 1)} &= \norm{\sigma^{(i\cdot k + 1)}}_1 - \lambda \cdot \paren*{\norm{\vec{s}^{(i \cdot k + 1)}}_1 - \norm{\sigma^{(i \cdot k + 1)}}_1} = \sum_{j=0}^{i-1} \beta^j - \lambda \cdot (k-1) \cdot \sum_{j=0}^{i-2} \beta^j \\
        &= \sum_{j=0}^{i-1} \paren*{\lambda (k-1)}^j - \lambda \cdot (k-1) \cdot \sum_{j=0}^{i-2} \paren*{\lambda (k-1)}^j = 1.
    \end{align*}
    Finally, if the gambler picks the last candidate(i.e. the randomized candidate), then their expected utility is 
    \begin{align*}
        \E_{\sigma \sim F} \sq*{U_{g_b}(\sigma, wk+ 1)} &= \epsilon \cdot \paren*{(1 - \epsilon)(1 + \lambda) \cdot \sum_{i = 0}^{w - 1} \beta^{i}/\epsilon - \lambda \cdot (k-1) \cdot \sum_{i=0}^{w-1} \beta^i} + (1 - \epsilon) \cdot \paren*{-\lambda \cdot k \cdot \sum_{i=0}^{w-1} \beta^i }\\
        &\leq (1 - \epsilon) \cdot \paren*{(1 + \lambda) \cdot \sum_{i = 0}^{w - 1} \beta^{i} - \lambda \cdot k \cdot \sum_{i=0}^{w-1} \beta^i } \\
        &= (1 - \epsilon) \cdot (1 - \lambda (k-1)) \cdot \frac{1 - \paren*{\lambda(k-1)}^w}{1 - \lambda (k-1)}  = (1 - \epsilon) \cdot \paren*{1 - \paren*{\lambda(k-1)}^w} < 1. 
    \end{align*}
    Therefore it is not possible for the biased gambler to get utility more than $1$, consequently, $\E_{\dist} \sq*{U^*_{g_b}} \leq 1$. On the other hand, the rational prophet can get expected utility
    \begin{align*}
        \E_{\dist}\sq*{U_{p_r}} = \E\sq*{V^*} &= \epsilon \cdot \norm{\sigma^{(wk+ 1)}}_1 + (1 - \epsilon) \cdot \norm{\sigma^{(wk)}}_1 \\
        &= \epsilon \cdot \paren*{(1  - \epsilon)(1 + \lambda) \cdot \sum_{i = 0}^{w - 1} \beta^{i}/\epsilon} + (1 - \epsilon) \cdot \sum_{i = 0}^{w - 1} \beta^{i} \\
        &= (1 - \epsilon) (2 + \lambda) \cdot \frac{1 - \paren*{\lambda(k-1)}^w}{1 - \lambda (k-1)}. 
    \end{align*}
    Similarly, the rational gambler can simply always choose the last candidate, thus
    \begin{align*}
        \E_{\dist}\sq*{U^*_{g_r}} \geq  \epsilon \cdot \paren*{(1  - \epsilon)(1 + \lambda) \cdot \sum_{i = 0}^{w - 1} \beta^{i}/\epsilon} = (1 - \epsilon) (1 + \lambda) \cdot \frac{1 - \paren*{\lambda(k-1)}^w}{1 - \lambda (k-1)}. 
    \end{align*} 
    Let $\epsilon$ be such that $(1 - \epsilon)^2 = 1 - \delta$ and let $w = \log_{\lambda (k-1)} \epsilon$. Then $(\lambda(k-1))^w = \epsilon$ and thus 
    \begin{align*}
        \frac{\E_{\dist}\sq*{U_{p_r}}}{\E_{\dist} \sq*{U^*_{g_b}}} \geq \frac{(2 + \lambda)\cdot(1 - \epsilon)^2}{1 - \lambda(k-1)} = \frac{2 + \lambda}{1 - \lambda(k-1)} \cdot (1 - \delta). 
    \end{align*}
    Similarly, 
    \begin{align*}
        \frac{\E_{\dist}\sq*{U^*_{g_r}}}{\E_{\dist} \sq*{U^*_{g_b}}} \geq \frac{(1 + \lambda)\cdot(1 - \epsilon)^2}{1 - \lambda(k-1)} = \frac{1 + \lambda}{1 - \lambda(k-1)} \cdot (1 - \delta). 
    \end{align*}
\end{proof}

\section{Omitted Proofs of Section \ref{result:IID}} \label{relaxed}
\claimIIDExpUtilityGap*

\begin{proof}
    The corollary follows in a relatively straightforward manner from claim~\ref{claim:expUtilityGap} and lemma~\ref{lem:detToiid}. Let $\sigma(m)$ be the sequence considered in claim~\ref{claim:expUtilityGap} of length $m$. Then $\frac{U_{p_r}(\sigma(m))}{U_{g_b}(\sigma(m))} = \Omega\paren*{(\lambda(k-1))^{\frac{m}{k}-1}}$, and in particular $\frac{U_{p_r}(\sigma(m))}{U_{g_b}(\sigma(m))} \geq (\lambda(k-1))^{\frac{m}{k}-1}$. Observe that $\sigma(m)$ is succinct. Thus by lemma~\ref{lem:detToiid}, for any $\epsilon \in (0, 1)$, there exists a distribution $\dist$ with $n = m^{\alpha (m-1)} \cdot \log^{\alpha}(m)$ where $\alpha = \log_{m}\paren*{\frac{(\lambda(k-1))^{\frac{m}{k}-1}}{\epsilon }} + 2$ such that $\frac{\E_{\dist} \sq*{U_{p_r}}}{\E_{\dist}\sq*{U^*_{g_b}}} \geq (1 - \epsilon) \cdot (\lambda(k-1))^{\frac{m}{k} - 1}$. Now in order to get a bound in terms of $n$ we only need to establish the relationship between $n$ and $m$. 
    
    Since $n = m^{\alpha (m-1)} \cdot \log^{\alpha}(m) \leq m^{m \cdot \alpha} = e^{\log{m} \cdot m \cdot \alpha}$, this means that  
    \begin{align*}
        \log n &\leq m \cdot \log m \cdot \paren*{\log_{m}\paren*{\frac{ (\lambda(k-1))^{\frac{m}{k}-1}}{\epsilon }} + 2} \\
        &\leq m \cdot \log\paren*{\frac{ (\lambda(k-1))^{\frac{m}{k}-1}}{\epsilon }} + 2 \cdot m \cdot \log m \\
        &\leq m \cdot \paren*{\frac{m}{k}-1} \cdot \log (\lambda(k-1)) + m \cdot \log\paren*{\frac{1}{ \epsilon}} + 2 m \cdot \log m. 
    \end{align*}
    Let $\epsilon = 1/2$, since $m \geq 2$, we know that $\log m \geq 1$, this means that 
    \begin{align*}
        \log n &\leq \frac{m^2}{k} \cdot \log (\lambda(k-1)) + m \cdot \log 2 + 2 m \log m \\
        &\leq \frac{m^2}{k} \cdot \log (\lambda(k-1)) + 4 m \log m\\
        &\leq 2 \cdot \max \paren*{\frac{m^2}{k} \cdot \log (\lambda(k-1)),  4 m \log m}.
    \end{align*}
    Now there are two cases: 
    \begin{enumerate}
        \item 
        When $\frac{m^2}{k} \cdot \log (\lambda(k-1)) \geq 4 m \log m$, we know that $\log n \leq 2 \cdot \frac{m^2}{k} \cdot \log (\lambda(k-1))$. This implies that $m^2 \geq \frac{k \cdot \log n}{2 \cdot \log (\lambda(k-1))}$ and thus $m \geq \frac{\sqrt{k} \cdot \log^{\frac{1}{2}} n }{\sqrt{2} \cdot \log^{1/2} (\lambda(k-1))}$. 
        \item
        When $\frac{m^2}{k} \cdot \log (\lambda(k-1)) < 4 m \log m$, we know that $\log n \leq 8 m \log m \leq 8 m^2$, which means that $m \geq \frac{\log^{\frac{1}{2}} n}{2 \sqrt{2}}$. 
    \end{enumerate}
    We conclude that no matter which case we are in, $m \geq \frac{1}{\sqrt{2}} \cdot \log^{\frac{1}{2}} n \cdot \min \paren*{\frac{\sqrt{k}}{\log^{1/2} (\lambda(k-1))}, \frac{1}{2}}$, which means that 
    \begin{align*}
        \frac{\E_{\dist} \sq*{U_{p_r}}}{\E_{\dist}\sq*{U^*_{g_b}}} \geq (1 - \epsilon) \cdot \paren*{\lambda(k-1))^{\frac{m}{k} - 1}} = \Omega \paren*{ 
            (\lambda(k-1))^{\frac{1}{\sqrt{2k}} \cdot \log^{\frac{1}{2}} n \cdot \min \paren*{\frac{1}{\log^{1/2} (\lambda(k-1))}, \frac{1}{2\sqrt{k}}} - 1}
        }.
    \end{align*}
    Let $f(\lambda,k) = \frac{1}{\sqrt{2k}} \cdot \min \paren*{\frac{1}{\log^{1/2} (\lambda(k-1))}, \frac{1}{2\sqrt{k}}} - 1$, and we get the desired claim that $\frac{\E_{\dist} \sq*{U_{p_r}}}{\E_{\dist}\sq*{U^*_{g_b}}} = \Omega \paren*{(\lambda(k-1))^{f(\lambda, k) \cdot \log^{\frac{1}{2}} n }}$. By claim~\ref{claim:classical},  $\frac{\E_{\dist} \sq*{U^*_{g_r}}}{\E_{\dist}\sq*{U^*_{g_b}}} = \Omega \paren*{(\lambda(k-1))^{f(\lambda, k) \cdot \log^{\frac{1}{2}} n }}$ as well.
\end{proof}

\claimIIDlogUtilityGap*
\begin{proof}
    The corollary follows in a relatively straightforward manner from claim~\ref{claim:linearUtilityGap} and lemma~\ref{lem:detToiid}. Let $\sigma(m)$ be the sequence considered in claim~\ref{claim:linearUtilityGap} of length $m$. Then $\frac{U_{p_r}(\sigma(m))}{U_{g_b}(\sigma(m))} = \Omega\paren*{\frac{m}{k}}$, and in particular $\frac{U_{p_r}(\sigma(m))}{U_{g_b}(\sigma(m))} \geq \frac{m}{k}$. Observe that $\sigma(m)$ is succinct. Thus by lemma~\ref{lem:detToiid}, for any $\epsilon \in (0, 1)$, there exists a distribution $\dist$ with $n = m^{\alpha (m-1)} \cdot \log^{\alpha}(m)$ where $\alpha = \log_{m}\paren*{\frac{m}{k \cdot \epsilon }} + 2$ such that $\frac{\E_{\dist} \sq*{U^*_{p_r}}}{\E_{\dist}\sq*{U_{g_b}}} \geq  (1 - \epsilon) \cdot \frac{m}{k}$.
    
    By using a similar analysis as in corollary~\ref{claim:iidExpUtilityGap}, we can compute the relationship between $n$ and $m$ as follows: 
    \begin{align*}
        \log n \leq \log \paren*{m^{\alpha (m-1)} \cdot \log^{\alpha}(m)} \leq \log \paren*{m^{\alpha m }} = \alpha \cdot m \log m.  
    \end{align*}
    Now let us set $\epsilon = 1/2$ and $\beta = \alpha - 3$, then
    \begin{align*}
       \beta = \paren*{\log_{m}\paren*{\frac{m}{k \cdot  \epsilon}} + 2} - (2 + \log_{m} m) =  \log_{m}\paren*{\frac{1}{kx\cdot \epsilon}} = \log_{m}\paren*{\frac{2}{k}} = O(1).
    \end{align*}
    
    We conclude that $\alpha = 3 + \beta = 3 + O(1) = \Theta(1)$. We will now show that $m = \Omega(\frac{\log n}{\log \log n})$. 
    \begin{enumerate}
        \item 
        Assume $m \geq \log n$, then it's clear that $m = \Omega(\frac{\log n}{\log \log n})$. 
        \item 
        Assume $m \leq \log n$, then $\log m \leq \log \log n$, since $\log n \leq \alpha \cdot m \log m$, we conclude that $$m \geq \frac{\log n}{\alpha \log m} \geq  \frac{\log n}{\alpha \log \log n} = \Omega\paren*{\frac{\log n}{\log \log n}}. $$
    \end{enumerate}
    Consequently, 
    $\frac{\E_{\dist} \sq*{U_{p_r}}}{\E_{\dist}\sq*{U^*_{g_b}}}  \geq (1 - \epsilon) \cdot \frac{m}{k} = \Omega\paren*{\frac{\log n}{k \cdot \log \log n}}
    $ and by claim~\ref{claim:classical} $ \frac{\E_{\dist} \sq*{U^*_{g_r}}}{\E_{\dist}\sq*{U^*_{g_b}}} = \Omega\paren*{\frac{\log n}{k \cdot \log \log n}}$ as well. 
\end{proof}

\thmIIDInftyGap*
\begin{proof}
    By corollaries~\ref{claim:iidExpUtilityGap} and ~\ref{claim:iidlogUtilityGap} we know that when $\lambda \cdot (k-1) \geq 1$, for any $n \in \mathbb{N}^+$, there exists an i.i.d distribution $\dist_0(n)$ such that the prophet and online utility ratio of $\dist_0(n)$ is an unboundedly increasing function in terms of the number of candidates $n$. Thus for any $T \in\mathbb{N}$, there exists a $n \in \mathbb{N}^+$ such that the prophet and online utility of $\dist_0(n)$ both exceed $T$. 
    
    Recall that the formal definition for competitive ratio is slightly different in the i.i.d prior distribution setting in contrast to the adversarial order setting. Namely, the prophet and online competitive ratios are defined as $\max \limits_{\dist \in \mathcal{U}}  \frac{\E_{\dist}\sq*{U_{p_r}}}{\E_{\dist} \sq*{U^*_{g_b}}}$ and $\max \limits_{\dist \in \mathcal{U}}  \frac{\E_{\dist}\sq*{U^*_{g_r}}}{\E_{\dist} \sq*{U^*_{g_b}}}$, where $\mathcal{U}$ is the set of all i.i.d distribution. 
    Observe that the prophet competitive ratio 
    $\max \limits_{\dist \in \mathcal{U}}  \frac{\E_{\dist}\sq*{U_{p_r}}}{\E_{\dist} \sq*{U^*_{g_b}}}$ is at least $\max \limits_{n \in \mathbb{N}^+} \frac{\E_{\dist_0(n)}\sq*{U_{p_r}}}{\E_{\dist_0(n)} \sq*{U^*_{g_b}}}$, thus any $T \in \mathbb{N}$ serves as a lower bound to the prophet competitive ratio, and thus the prophet competitive ratio is equal to $\infty$. Similarly, the online competitive ratio is equal to $\infty$. 
\end{proof} 
\claimUBBoundedUtilityIID*

\begin{proof}
    We will consider a sequence that is similar to the mostly deterministic distribution considered in theorem~\ref{thm:UBBoundedUtility}, except that the last candidate, which has randomized value, is excluded. Specifically, let $\beta = \lambda \cdot k$, let $w$ be a parameter that we will set later and let $\sigma$ be the following sequence:
    \begin{center}
        \begin{tabular}{ l l l l }
            $(1, 0, \cdots, 0)$, &$(0, 1, 0, \cdots, 0)$, &$\cdots$, &$(0, \cdots, 0, 1)$,\\
            $(1 + \beta, 0, \cdots, 0)$, &$(0, 1 + \beta, 0, \cdot(1 - \lambda (k-1))(1 + \delta)s, 0)$, &$\cdots$, &$(0, \cdots, 0, 1 + \beta)$,\\
            $(1 + \beta + \beta^2, 0, \cdots, 0)$, &$(0, 1 + \beta + \beta^2, 0, \cdots, 0)$, &$\cdots$, &$(0, \cdots, 0, 1 + \beta + \beta^2)$,\\
            &$\cdots$ & &\\
            $(\sum_{i = 0}^{w - 1} \beta^{i}, 0 , \cdots, 0)$, &$(0, \sum_{i = 0}^{w - 1} \beta^{i}, 0, \cdots, 0)$, &$\cdots$, &$(0, \cdots, 0, \sum_{i = 0}^{w - 1} \beta^{i})$. \ 
        \end{tabular}
    \end{center}
    As reasoned in theorem~\ref{thm:UBBoundedUtility}, the biased gambler has utility $\leq 1$ when selecting any candidate in this sequence. However, the rational prophet will pick a candidate in the last row, and gains utility $\sum_{i = 0}^{w-1} \beta^{i} = \frac{1 - \beta^w}{1 - \beta}$. Thus, 
    \begin{align*}
        \frac{U_{p_r}(\sigma)}{U_{g_b}(\sigma)} \geq \frac{1 - \beta^w}{1 - \beta} = \frac{1 - (\lambda(k-1))^w}{1 - \lambda(k-1)}. 
    \end{align*}
   Let $\epsilon$ be such that $(1 - \epsilon)^2 = 1 - \delta$, and let $w = \epsilon \cdot \log_{\lambda(k+1)}$, then $(\lambda(k-1))^w = \epsilon$. Thus
   \begin{align*}
    \frac{U_{p_r}(\sigma)}{U_{g_b}(\sigma)} \geq \frac{1 - (\lambda(k-1))^w}{1 - \lambda(k-1)} \geq (1 - \epsilon) \cdot \frac{1}{1 - \lambda(k-1)}. 
\end{align*} 
    Since all candidates in $\sigma$ are distinct, $\sigma$ is succinct, hence by lemma~\ref{lem:detToiid}, there exists a prior distribution $\dist$ where all candidates are i.i.d such that 
    \begin{align*}
        \frac{\E_{\dist}[U_{p_r}]}{\E_{\dist}[U^*_{g_b}]} \geq (1 - \epsilon) \cdot \frac{U_{p_r}(\sigma)}{U_{g_b}(\sigma)}  = (1 - \epsilon)^2 \cdot  \frac{1}{1 - \lambda(k-1)}  = (1 - \lambda(k-1)) \cdot (1 - \delta). 
    \end{align*}
\end{proof}

\section{Monotonicity Results}
\label{app:monotonicity}

\begin{claim}
If $\|s'^{(t)}\|_1 < \|s^{(t)}\|_1$ for sequences $\sigma'$ and $\sigma$, where $\sigma^{(t)} = \sigma'^{(t)}$, then $U_{g_b}(\sigma, t) \leq U_{g_b}(\sigma', t)$. 
\end{claim}
\begin{proof}
From the definition of the biased gambler's utility we have: 
$$\|\sigma^{(t)}\|_{1} - \lambda \bigg(\|\vec{s}^{(t)}\|_1 - \|\sigma^{(t)}\|_1\bigg) \leq \|\sigma'^{(t)}\|_{1} - \lambda \bigg(\|\vec{s}'^{(t)}\|_1 - \|\sigma'^{(t)}\|_1\bigg)  \implies U_{g_b}(\sigma, t) \leq U_{g_b}(\sigma', t) $$
\end{proof}

\begin{claim}
For any super candidate $s^{(t)}$ and candidate $\sigma^{(t)}$, if $\lambda' > \lambda \geq 0$, then, when they select candidate $\sigma^{(t)}$, the biased gambler receives greater utility with loss aversion $\lambda$ than with loss aversion $\lambda'$, namely $U_{g_{b_{\lambda}}}(\sigma, t) \geq U_{g_{b_{\lambda'}}}(\sigma, t)$.
\end{claim}
\begin{proof}
From the definition of the biased gambler's utility and the fact that $\lambda' > \lambda$, we have: 
$$\|\sigma^{(t)}\|_{1} - \lambda \bigg(\|\vec{s}^{(t)}\|_1 - \|\sigma^{(t)}\|_1\bigg) \geq \|\sigma^{(t)}\|_{1} - \lambda' \bigg(\|\vec{s}^{(t)}\|_1 - \|\sigma^{(t)}\|_1\bigg)  \implies U_{g_{b_{\lambda}}}(\sigma, t) \geq U_{g_{b_{\lambda'}}}(\sigma, t)$$
\end{proof}

Next, we conduct analysis regarding \citet{kleinberg2021optimal}'s construct of patience in our multidimensional setting. For sake of semantic clarity, we re-iterate their definition here: 
\begin{definition}[Patient]
A stopping rule $\pi$ is more patient than a stopping rule $\pi '$ if for every realization of the candidates the stopping rule $\pi$ either selects the same candidate as $\pi'$ or selects a candidate that is later in the sequence.
\end{definition}
Furthermore, for any stopping rule $\pi$, let us define the notation $\pi(\sigma)$ to denote the value of the candidate selected by the $\pi$ when the realization of the candidates is the sequence $\sigma$. Specifically, $\pi_\lambda(\sigma)$ denotes the value of the candidate selected by the optimal $\lambda$-biased stopping rule.  

\begin{claim}
\label{patience_expectation}
Given a prior distribution $\dist$, if an optimal multidimensional $\lambda$-biased stopping rule $\pi_{\lambda}$ is more patient than another multidimensional stopping rule $\pi$, then the expected value of the candidate selected by $\pi_{\lambda}$ is higher than the expected value of the candidate selected by $\pi$ (i.e. $\E_{\sigma \sim \dist}\sq*{\pi_{\lambda}(\sigma)} \geq \E_{\sigma \sim \dist}\sq*{\pi(\sigma)}$).
\end{claim}
\begin{proof}
First, note that we are reasoning about the objective value of the selected candidates not about the utilities that they yield to the agents. Let $t(\sigma)$ and $w(\sigma)$ denote the index of the candidate selected by $\pi$ and $\pi_\lambda$ when the realization of candidates is the sequence $\sigma$.   

By the definition of patience, for any sequence $\sigma$ in the support of $\dist$, $t(\sigma) \leq w(\sigma)$. Since $\pi_\lambda$ is optimal for a $\lambda$-biased agent, for a $\lambda$-biased agent playing optimally on $\sigma$:
$$\E_{\sigma \sim \dist}\sq*{U_{g_b}(\sigma, w(\sigma))} \geq \E_{\sigma \sim \dist}\sq*{U_{g_b}(\sigma, t(\sigma))} $$ 
Note that this is a statement about utilities not about the objective values of the items selected. Now we can substitute in the definition of $U_{g_b}$ and relate the values of candidates $\sigma^{(t(\sigma))}$ and $\sigma^{(w(\sigma))}$. 
\begin{align*}
   \E_{\sigma \sim \dist}\sq*{\norm{\sigma^{(w(\sigma))}}_1  - \lambda(\norm{s^{(w(\sigma))}}_1 - \norm{\sigma^{(w(\sigma))}}_1)} 
   \geq 
   \E_{\sigma \sim \dist} \sq*{\norm{\sigma^{(t(\sigma))}}_1  - \lambda(\norm{s^{(t(\sigma))}}_1 - \norm{\sigma^{(t(\sigma))}}_1)}.
\end{align*}
Re-expressing this algebraically, we have: 
\begin{align*}
    (1 + \lambda) \E_{\sigma \sim \dist}\sq*{\norm{\sigma^{(w(\sigma))}}_1}  - \lambda \E_{\sigma \sim \dist}\sq*{\norm{s^{(w(\sigma))}}_1} 
    \geq 
    (1 + \lambda)  \E_{\sigma \sim \dist}\sq*{\norm{\sigma^{(t(\sigma))}}_1}  - \lambda \E_{\sigma \sim \dist}\sq*{\norm{s^{(t(\sigma))}}_1}.  
\end{align*}

Now assume for contradiction that $\E_{\sigma \sim \dist}\sq*{\norm{\sigma^{(w(\sigma))}}_1} < \E_{\sigma \sim \dist}\sq*{\norm{\sigma^{(t(\sigma))}}_1}$. Observe that since for any sequence $\sigma$ in the support of $\dist$, $w(\sigma) \geq t(\sigma)$, it follows that $\|s^{(t(\sigma))}\|_1 \leq \|s^{(w(\sigma))}\|_1$ and thus $\E_{\sigma \sim \dist}\sq*{\norm{s^{(t(\sigma))}}_1} \leq \E_{\sigma \sim \dist}\sq*{\norm{s^{(w(\sigma))}}_1}$.  Consequently, our assumption violates the above inequality which we know to be true, thus it must be the case that $\E_{\sigma \sim \dist}\sq*{\norm{\sigma^{(w(\sigma))}}_1} > \E_{\sigma \sim \dist}\sq*{\norm{\sigma^{(t(\sigma))}}_1}$.
\end{proof}

\begin{claim}
\label{monotonic_patience}
The optimal $\lambda$-biased stopping rule, $\pi_{\lambda}$, is more patient than the optimal $\lambda'$-biased stopping rule, $\pi_{\lambda'}$, for $\lambda \leq \lambda'$.
\end{claim}
\begin{proof}
\newcommand{\event}{A}
Let $t'$ denote the random variable representing the index of the candidate $\pi_\lambda'$ selects. For contradiction, assume that there exists realization $\sigma^{(1)}, \cdots, \sigma^{(t)} = \sigma_1,....,\sigma_t$ for which $\pi_{\lambda}$ selects $\sigma^{(t)}$ and $t'$ (the choice of $\pi_{\lambda'}$) is the index of some candidate after $t$. Let event $\event$ be the event where $\sigma^{(1)}, \cdots, \sigma^{(t)} = \sigma_1,....,\sigma_t$ and let $s_t$ be the super candidate at time step $t$ conditioned on event $\event$. We claim that after seeing $\sigma_1,....,\sigma_t$, the $\lambda'$-biased agent will get higher expected utility by using a strategy that selects $\sigma_t$ with probability $1$ than by using any other strategy. Thus our assumption must be false. We now prove our claim. Since $\pi_{\lambda}$ is optimal for a $\lambda$-biased agent, it must be the case that $\E_{\sigma \sim \dist} \cond*{U_{g_{b_{\lambda}}}(\sigma, t)}{\event} \geq \E_{\sigma \sim \dist}\cond*{U_{g_{b_{\lambda}}}(\sigma, t')}{\event}$, namely for a $\lambda$-biased agent, we have:
\begin{align*}
    \norm{\sigma_t}_1 - \lambda (\norm{s_t}_1 - \norm{\sigma_t}_1) \geq \E_{\sigma \sim \dist}\cond*{\|\sigma^{(t')}\|_1  - \lambda(\|s^{(t')}\|_1 - \|\sigma^{(t')} \|_1)}{\event}. \qquad (1)
\end{align*}

As an intermediate claim, We show that $
    \norm{s_t}_1 - \norm{\sigma_t}_1 \leq \E_{\sigma \sim \dist}\cond*{\norm{s^{(t')}}_1 - \norm{\sigma^{(t')}}_1}{\event}.$
First, let's consider the case where $\norm{\sigma_t}_1 \leq \E_{\sigma \sim \dist}\cond*{\norm{\sigma^{(t')}}_1}{\event}$. Due to the fact that the value of the super candidate is monotonically increasing, we have $\norm{s_t}_1 \leq \E_{\sigma \sim \dist}\cond*{\norm{s^{(t')}}_1}{\event}.$ Combining the two inequalities gives us the desired inequality. 

Next, we consider the case where $\norm{\sigma_t}_1 > \E_{\sigma \sim \dist}\cond*{\norm{\sigma^{(t')}}_1}{\event}$. Observe that $$\E_{\sigma \sim \dist}\cond*{\norm{s^{(t')}}_1}{\event} - \E_{\sigma \sim \dist}\cond*{\norm{\sigma^{(t')}]}_1}{\event} \geq  \norm{s_t}_1 - \E_{\sigma \sim \dist}\cond*{\norm{\sigma^{(t')}]}_1}{\event},$$ since $s_t \leq \E_{\sigma \sim \dist}\cond*{\|s^{(t')}\|_1}{\event}$. Combining this with the assumption that $\norm{\sigma_t}_1 > \E_{\sigma \sim \dist}\cond*{\norm{\sigma^{(t')}}_1}{\event}$, we get $\E_{\sigma \sim \dist}\cond*{\norm{s^{(t')}}_1}{\event} - \E_{\sigma \sim \dist}\cond*{\norm{\sigma^{(t')}]}_1}{\event} \geq \norm{s_t}_1 - \norm{\sigma_t}_1$, as desired. 

Having shown the intermediate claim, we can re-express it algebraically as follows:  
$$-(\lambda' - \lambda) \cdot (\norm{s_t}_1 - \norm{\sigma_t}_1) \geq -(\lambda' - \lambda) \cdot \E_{\sigma \sim \dist}\cond*{(\|s^{(t')}\|_1 - \|\sigma^{(t')}\|_1)}{\event}. \qquad (2)$$
Combining statements (1) and (2), we get the following:
\begin{align*}
    \norm{\sigma_t}_1 - \lambda' (\norm{s_t}_1 - \norm{\sigma_t}_1) \geq \E_{\sigma \sim \dist}\cond*{\|\sigma^{(t')}\|_1  - \lambda'(\|s^{(t')}\|_1 - \|\sigma^{(t')} \|_1)}{\event}. \qquad (3)
\end{align*}
By the definition of the biased agents utility, statement (3) is equivalent to the statement that $\E_{\sigma \sim \dist} \cond*{U_{g_{b_{\lambda'}}}(\sigma, t)}{\event} \geq \E_{\sigma \sim \dist}\cond*{U_{g_{b_{\lambda'}}}(\sigma, t')}{\event}$, namely that conditioned on event $\event$, a $\lambda'$-biased agent is at least as well off selecting $t$ as selecting $t'$ thus contradicting our initial assumption and showing no such $\sigma_1, \cdots, \sigma_t$ can exist. 
\end{proof}

\begin{corollary}
If $\lambda \leq \lambda'$, then, when both are playing optimally, the expected value of the candidate selected by a $\lambda$-biased agent is higher than the expected value of a candidate selected by a $\lambda'$-biased agent(i.e. $\E_{\sigma \sim \dist} \sq*{\pi_{\lambda}(\sigma)} \geq \E_{\sigma \sim \dist} \sq*{\pi_{\lambda'}(\sigma)}$). 
\end{corollary}

\begin{proof}
It follows as a direct consequence of claims  \ref{patience_expectation} and \ref{monotonic_patience} when they are taken together. 
\end{proof}
We now move on to consider monotonicity in the number of candidates.

\begin{claim}
\label{subset_patience}
Let $n, m$ be two integers such that $m > n$. Let $\dist_1, \cdots, \dist_{m}$ be $m$ value distributions for each of n+1 distinct candidates. Let $\dist = \times_{i \in [n]} \dist_i$ and $\dist' = \times_{i \in [m]} \dist_i$. We will use $\pi_\lambda$ and $\pi_\lambda'$ to denote the optimal $\lambda$-biased stopping rule for $\sigma \sim \dist$ and $\sigma' \sim \dist'$ respectively. We claim that $\pi_\lambda'$ is more patient than $\pi_\lambda$. 
\end{claim}

\begin{proof}
    Assume for the sake of contradiction that there exist $\sigma_1, \cdots, \sigma_{m}$ where $\sigma = \sigma_1, \cdots, \sigma_{n}$ and $\sigma' = \sigma_1, \cdots, \sigma_{m}$ such that $\pi_{\lambda}(\sigma)$ is later in the sequence than $\pi_{\lambda}'(\sigma')$. Since, $\pi_{\lambda}$ is the optimal $\lambda$-biased stopping rule, this assumption implies that $$U_{g_{b_{\lambda}}}(\sigma, \pi_{\lambda}(\sigma)) > U_{g_{b_{\lambda}}}(\sigma, \pi_{\lambda}'(\sigma')) = U_{g_{b_{\lambda}}}(\sigma', \pi_{\lambda}'(\sigma')).$$ This is a contradiction since a stopping rule for $\sigma' \sim \dist'$ can get strictly better utility compared to $\pi_{\lambda}'$ by taking the same action as $\pi_{\lambda}$ and never selecting candidates with index greater than $n$.
\end{proof}

\begin{claim}
Let $n, m$ be two integers such that $m > n$. Let $\dist_1, \cdots, \dist_{m}$ be $m$ value distributions for each of $m$ distinct candidates. Let $\dist = \times_{i \in [n]} \dist_i$ and $\dist' = \times_{i \in [m]} \dist_i$. Consider $\sigma \sim \dist$ and $\sigma' \sim \dist'$ where $\sigma$ is a prefix of $\sigma'$. For any $\lambda \geq 0$ and optimal $\lambda$-biased stopping rule $\pi_{\lambda}$ and $\pi_{\lambda}'$ for $\sigma \sim \dist$ and $\sigma' \sim \dist'$ respectively, $U_{g_{b_{\lambda}}}(\sigma, \pi_{\lambda}(\sigma)) \leq U_{g_{b_{\lambda}}}(\sigma', \pi_{\lambda}'(\sigma'))$ and $\pi_{\lambda}(\sigma)\leq \pi_{\lambda}'(\sigma')$. Semantically, both the utility and expected value of the selected candidate improve monotonically as more candidates are added to the end of the sequence. 
\end{claim}
\begin{proof}
    Since $\sigma$ is a prefix of $\sigma'$, an agent can simply select candidate $\pi_{\lambda}(\sigma)$ and receive at least as much utility on $\sigma'$ as on $\sigma$, thus $U_{g_{b_{\lambda'}}}(\sigma, \pi_{\lambda}(\sigma)) \leq U_{g_{b_{\lambda'}}}(\sigma', \pi_{\lambda}'(\sigma'))$. By claim \ref{subset_patience}, the agent is more patient on $\sigma'$ than on $\sigma$; combining this with claim \ref{patience_expectation}, it must be the case that the expected value of the candidate selected on $\sigma'$ is greater than or equal to the expected value of that selected on $\sigma$. 
\end{proof}

Following the outline of \citet{kleinberg2021optimal}, we now consider adding a candidate at the beginning of the sequence. Such addition can decrease the expected utility of a $\lambda$-biased agent with comparative loss aversion, whereas it cannot decrease that of a rational agent.

\begin{claim}
Let $\dist_1, \cdots, \dist_{n+1}$ be $n+1$ value distributions for each of $n+1$ distinct candidates. Let $\dist = \times_{i =2}^{n+1} \dist_i$ and $\dist' = \times_{i \in [n+1]} \dist_i$. For any $\lambda \geq 0$ and optimal $\lambda$-biased stopping rule $\pi_{\lambda}$ and $\pi_{\lambda}'$ for $\dist$ and $\dist'$ respectively, consider sequences $\sigma \sim \dist$ and $\sigma' \sim \dist'$, where $\sigma' = \sigma^{(0)'} + \sigma$. (We can view $\sigma'$ as the result of appending a candidate, $\sigma^{(0)'}$, to the beginning of $\sigma$.) We have $U_{g_b}(\sigma', \pi_{\lambda}'(\sigma')) \geq \frac{1}{1 + \lambda} U_{g_b}(\sigma, \pi_{\lambda}(\sigma))$.
\end{claim}
\begin{proof}
We lower bound $U_{g_b}(\sigma', \pi_{\lambda}(\sigma'))$ by considering the utilities for two possible stopping rules:
\begin{enumerate}
    \item Accept the first candidate in $\sigma'$, namely $\sigma^{(0)'}$
    \item Ignore the first candidate, $\sigma^{(0)'}$ and instead accept the same candidate selected on $\sigma$, namely candidate $\pi_{\lambda}(\sigma)$. 
\end{enumerate}
In the first case, since the first candidate seen is selected, there is no regret and the utility of the gambler is $\|\sigma^{(0)'} \|_1$. In the second case, since the same candidate is chosen as under $\sigma$ and the super candidate can increase by at most the value of the candidate appended to the beginning, the utility is lower bounded by 
$U_{g_b}(\sigma, \pi_{\lambda}(\sigma)) - \lambda \|\sigma^{(0)'}\|_1$. Since $\pi_{\lambda}$ is the utility maximizing strategy we thus have:
$$U_{g_b}(\sigma', \pi_{\lambda}(\sigma')) \geq max \{U_{g_b}(\sigma, \pi_{\lambda}(\sigma)) - \lambda \|\sigma^{(0)'}\|_1,  \|\sigma^{(0)'} \| \} \qquad (1)$$

Consider the case where $U_{g_b}(\sigma, \pi_{\lambda}(\sigma)) - \lambda \|\sigma^{(0)'} \|_1 \leq  \|\sigma^{(0)'}\|_1 $, which algebraically re-expressed is $\|\sigma^{(0)'}\|_1 \geq \frac{U_{g_b}(\sigma, \pi_{\lambda}(\sigma))}{1 + \lambda}$. Combining this with inequality $(1)$ yields $U_{g_b}(\sigma', \pi_{\lambda}(\sigma')) \geq \frac{U_{g_b}(\sigma, \pi_{\lambda}(\sigma))}{1 + \lambda}$. Now consider the case where $U_{g_b}(\sigma, \pi_{\lambda}(\sigma)) - \lambda \|\sigma^{(0)'}\|_1 >  \|\sigma^{(0)'}\|_1 $, algebraically re-expressed this is the case where  $\|\sigma^{(0)'}\|_1 < \frac{U_{g_b}(\sigma, \pi_{\lambda}(\sigma))}{1 + \lambda}$. Here we get:
$$U_{g_b}(\sigma', \pi_{\lambda}(\sigma')) \geq U_{g_b}(\sigma, \pi_{\lambda}(\sigma)) - \lambda \|\sigma^{(0)'}\|_1 \geq U_{g_b}(\sigma, \pi_{\lambda}(\sigma)) - \lambda \frac{U_{g_b}(\sigma, \pi_{\lambda}(\sigma))}{1 + \lambda} = \frac{U_{g_b}(\sigma, \pi_{\lambda}(\sigma))}{1 + \lambda} $$
Thus, we have proven the first the claim, namely that $U_{g_b}(\sigma', \pi_{\lambda}'(\sigma')) \geq \frac{U_{g_b}(\sigma, \pi_{\lambda}(\sigma))}{1 + \lambda}$. 
\end{proof}

\end{document}